\pdfoutput=1

\documentclass[11pt]{article}
\usepackage[utf8]{inputenc}
\usepackage{enumitem}
\usepackage[skins,xparse,breakable]{tcolorbox}
\usepackage{amsmath,amsfonts,amssymb}

\usepackage{bbm}
\usepackage{dsfont}
\usepackage{amsthm}
\usepackage{fullpage}

\newtheorem{theorem}{Theorem}[section]
\newtheorem{corollary}{Corollary}[theorem]
\newtheorem{lemma}[theorem]{Lemma}
\newtheorem{proposition}[theorem]{Proposition}
\theoremstyle{definition}
\newtheorem{definition}{Definition}

\newtheorem{task}{Task}
\newtheorem{fact}{Fact}
\newtheorem{remark}{Remark}
\usepackage{thmtools} 
\usepackage{thm-restate}
\usepackage[margin=0.95in]{geometry}
\usepackage{authblk}
\usepackage{mathrsfs}  
\usepackage{mathtools}  
\usepackage{physics}  
\newtheorem{cor}[theorem]{Corollary}

\DeclareMathOperator{\poly}{poly}
\DeclareMathOperator{\perm}{Perm}
\DeclareMathOperator{\type}{type}
\DeclareMathOperator{\bintype}{bintype}
\DeclareMathOperator{\Dist}{Dist}
\DeclareMathOperator{\phase}{phase}

\usepackage{braket}

\newcommand{\lem}[1]{\hyperref[lem:#1]{Lemma~\ref*{lem:#1}}}
\newcommand{\defn}[1]{\hyperref[defn:#1]{Definition~\ref*{defn:#1}}}
\newcommand{\thmref}[1]{\hyperref[thmref:#1]{Theorem~\ref*{thmref:#1}}}
\newcommand{\factref}[1]{\hyperref[factref:#1]{Fact~\ref*{factref:#1}}}

\newcommand{\rf}{\mathsf{r}_f}
\newcommand{\rp}{\mathsf{r}_p}

\newcommand{\m}{\mathsf{row}}
\newcommand{\n}{\mathsf{m_2}}
\newcommand{\1}{\mathsf{t}}
\newcommand{\Pisym}{\Pi_{\mathrm{sym}}}
\newcommand{\C}{\mathbb{C}}
\newcommand{\E}{\mathbb{E}}
\newcommand{\B}{\mathsf{B}}
\newcommand{\V}{\mathsf{V}}
\newcommand{\TD}{\mathsf{TD}}
\newcommand{\MPS}{\mathsf{MPS}}

\newcommand{\eqn}[1]{(\ref{eqn:#1})}

\title{Quantum Pseudoentanglement}
\date{\today}

\author[1]{Scott Aaronson\thanks{aaronson@cs.utexas.edu }}
\author[2]{Adam Bouland\thanks{abouland@stanford.edu}}
\author[3]{Bill Fefferman\thanks{wjf@uchicago.edu}}
\author[3]{Soumik Ghosh\thanks{soumikghosh@uchicago.edu}}
\author[4]{Umesh Vazirani\thanks{vazirani@eecs.berkeley.edu}}
\author[2]{Chenyi Zhang\thanks{chenyiz@stanford.edu}}
\author[2]{Zixin Zhou\thanks{jackzhou@stanford.edu}}

\affil[1]{\normalsize{Department of Computer Science, University of Texas, Austin}}
\affil[2]{\normalsize{Department of Computer Science, Stanford University}}
\affil[3]{Department of Computer Science, University of Chicago}
\affil[4]{\normalsize{Department of Electrical Engineering and Computer Sciences, University of California, Berkeley}}

\begin{document}

\maketitle
\thispagestyle{empty}

\begin{abstract}
Entanglement is a quantum resource, in some ways analogous to randomness in classical computation. Inspired by recent work of Gheorghiu and Hoban, we define the notion of ``pseudoentanglement'', a property exhibited by ensembles of efficiently constructible quantum states which are indistinguishable from quantum states with maximal entanglement. Our construction relies on the notion of quantum pseudorandom states -- first defined by Ji, Liu and Song -- which are efficiently constructible states indistinguishable from (maximally entangled) Haar-random states.  Specifically, we give a construction of pseudoentangled states with entanglement entropy arbitrarily close to $\log n$ across every cut, a tight bound providing an exponential separation between computational vs information theoretic quantum pseudorandomness. We  discuss applications of this result to Matrix Product State testing, entanglement distillation, and the complexity of the AdS/CFT correspondence.
As compared with a previous version of this manuscript (arXiv:2211.00747v1) this version introduces a new pseudorandom state construction, has a simpler proof of correctness, and achieves a technically stronger result of low entanglement across all cuts simultaneously.

\end{abstract}

\clearpage
\pagenumbering{arabic} 
\newpage

\section{Introduction}

Randomness is a resource in classical computation and cryptography, and the theory of pseudorandomness plays a central role in the study of this resource. Entanglement plays an analogous role and is a central resource in quantum information and computation. Inspired by the definition of pseudorandomness, and recent work of Gheorghiu and Hoban \cite{gheorghiu2020estimating}, we define the notion of pseudoentanglement. Informally we say that an ensemble of quantum states is pseudoentangled if the states are efficiently constructible and have small entanglement but are indistinguishable from quantum states with maximal entanglement. 

The study of quantum pseudoentanglement is closely related to the concept of quantum pseudorandom states introduced by Ji, Liu and Song \cite{Ji2018}. \emph{Pseudorandom states} are ensembles of quantum states which can be prepared by efficient quantum circuits, yet which masquerade as Haar-random states, even to arbitrary poly-time quantum algorithms using an arbitrary polynomial number of  copies of the state.
While such a strong form of pseudorandomness is impossible in the information-theoretic setting \cite{brandao2016local}, Ji, Liu and Song showed it is possible to construct such states in a \emph{computational} setting using post-quantum cryptography -- in particular using any quantum secure pseudorandom function, a standard cryptographic primitive \cite{zhandry2012construct}.
This notion has many applications in cryptography \cite{ananth2021cryptography,kretschmer2022quantum}, complexity theory \cite{kretschmer2021quantum}, and quantum gravity \cite{bouland2019computational,kim2020ghost}.

In this paper we give a new family of pseudorandom quantum states which have low entanglement rank (and therefore entropy) across every cut. A simple swap test argument shows that any pseudorandom quantum state must necessarily have $\omega(\log n)$ entanglement entropy across any cut \cite{Ji2018}. It was an open question to exhibit pseudorandom quantum states which saturate this entanglement entropy lower bound. Here we give a construction that is optimal and achieves entanglement entropy arbitrarily close to $\log n$ across every cut.  This should be contrasted with information theoretic notions of pseudorandomness, such as unitary $t$-designs, which require entanglement entropy $\Omega(n)$ across each cut\footnote{This is true in expectation for $t\geq 2$ \cite{smith2006typical,dahlsten2005exact} and becomes more concentrated as $t=4$ and higher \cite{low2009large,cotler2022fluctuations}.} , and consequently our results obtain an exponential separation between computational vs information theoretic quantum pseudorandomness. Moreover, since by definition pseudorandom states are indistinguishable from Haar random states, which have maximal entropy across every cut, this ensemble of states is also pseudoentangled. 

The construction of the pseudoentangled family is quite simple to describe. Let $S \subseteq \{0,1\}^n$ be a pseudorandom subset of superpolynomial support $|S| = s(n)$, and $f:\{0,1\}^n \rightarrow \{0,1\}$ a PRF. Then we prove that the state
\[ \displaystyle\sum_{x\in S} (-1)^{f(x)} \ket{x}\]
is pseudorandom. The Schmidt rank across any cut is bounded by $s(n)$ and therefore the entanglement entropy is bounded by $\log s(n)$, and for $s(n)$ to be superpolynomial it is only necessary for the entanglement entropy to grow faster than $\log n$. Showing that this family of quantum states is efficiently preparable and pseudorandom therefore establishes that it is indistinguishable from the maximally entangled Haar random states, and is therefore a pseudoentangled family of states. 

\begin{restatable}[Pseudorandom states with low entanglement across all cuts (Informal)]{thm}{Lowentanglement}
\label{Thm:Lowentanglement}
 For any function $f(n)=\omega(\log n)$, there exists ensembles of pseudorandom states with entanglement entropy $\Theta(f(n))$ across all cuts of the state\ simultaneously \footnote{Technically, this is across all cuts of the state where one size of the partition is of size $\Omega(f(n))$, as the statement is trivially false otherwise.}. 
\end{restatable}

We note that a previous version of this result appeared on the arXiv with identification number arXiv:2211.00747v1 and as a contributed talk at QIP 2023. 
Our prior construction was based on the random phase state construction of \cite{Ji2018,brakerski2019pseudo}, and we showed it is possible to decrease the entanglement to any $f(n)=\omega(\log n)$ across a single fixed cut, which we include in Appendix \ref{app:singlecutjls}.
It turns out it is possible to generalize our prior construction to have low entanglement with respect to additional cuts.  We achieve this by repeatedly using the same technique to reduce the entanglement across certain cuts without accidentally blowing up the entanglement across other cuts.  
While this allows us to produce, for example, 1D pseudorandom states with ``pseudo-area law'' scaling of entanglement (i.e., the entanglement of any cut is upper bounded by $A\cdot \poly\log(n)$ where $A$ is the area of the cut when the qubits are arranged on a line) which we prove in Appendix \ref{app:arealawjls}, the technique requires a careful choice of cuts and does not give us the ability to reduce entanglement across all cuts.
Compared to that result, our current construction introduces a new pseudorandom state construction, has a simpler proof of correctness, and achieves a technically stronger result of low entanglement across all cuts simultaneously.  As we will show, the pseudo-area law scaling of entanglement allows us to prove strong property testing lower bounds, such as for testing Matrix Product States.

\subsection{Applications}
\label{sec:applicationsintro}
Given the central role played by pseudorandomness in classical computer science, we expect that the notion of pseudoentanglement will shed new light on our understanding of quantum entanglement.  Here we scratch the surface by providing some initial applications.  

First, our main result implies  new lower bounds in property testing. For example, suppose one wishes to tell if an $n$-qubit state has a Matrix Product State (MPS) description of bond dimension $k$, or is far from any such state?
This is the ``MPS-testing'' problem.
Soleimanifar and Wright \cite{soleimanifar2022testing} recently showed MPS testing requires $\Omega(\sqrt{n})$ copies of the state in an info-theoretic sense.
We show that MPS testing requires $\Omega\big(\sqrt{k}\big)$ copies of the state, either in info-theoretic or computational settings, as a corollary of our low-entropy PRS construction. While  incomparable to the Soleimanifar-Wright bound, this is a stronger lower bound in the regime of high bond dimension $k$ -- so is saying that it gets more and more difficult to determine if a state is an MPS as the bond dimension grows.  We describe this application in more detail in Section \ref{sec:mps}.

Another classic property testing problem is to estimate the Schmidt rank of many copies of an unknown quantum state \cite{montanaro2013survey}.
While it is known that in general this is a difficult problem \cite{childs2007weak}, prior lower bounds for this problem have relied on input quantum states which are not efficiently constructible. 
Our work implies that Schmidt rank testing remains intractable in the setting where states are efficiently constructible, and gives analogous lower bounds for a number of related property testing/tomography problems, such as estimating the largest Schmidt coefficients of an unknown state (see Section \ref{schmidt rank}).
In addition, our pseudoentanglement construction can be used to prove lower bounds on entanglement distillation protocols for extracting entanglement from Haar-entanglement random states via the Schur transform (see Section \ref{app:distillation}).

Finally our work has applications to quantum gravity theory.
A central theme of quantum gravity that entanglement is related to the geometry of general relativity, through dualities such as the AdS/CFT correspondence.
The construction of pseudoentangled states within these theories might give additional evidence that the duality maps must be exponentially difficult to compute \cite{bouland2019computational}, which was also part of Hoban and Gheorghiu's motivation for their work \cite{gheorghiu2020estimating}.
We discuss this further in Section \ref{subsec:adscft}.

\subsection{A formal definition of pseudo-entanglement}
\label{sec:pseudoentanglement}

We now proceed to define pseudoentanglement. A pseudoentangled state ensemble (PES) with gap $f(n)$ vs. $g(n)$ consists of two ensembles of $n$-qubit states $\ket{\Psi_k},\ket{\Phi_k}$ indexed by a secret key $k\in \{0,1\}^{\text{poly}(n)}$, with the following properties:
\begin{itemize}
    \item Given $k$, $\ket{\Psi_k}$ ($\ket{\Phi_k}$, respectively) is efficiently preparable by a uniform, poly-sized quantum circuit.
    \item  With probability at least $1 - \frac{1}{\text{poly}(n)}$ over the choice of $k$, the entanglement entropy across every cut of $\ket{\Psi_k}$ ($\ket{\Phi_k}$, respectively) is $\Theta(f(n))$ ($\Theta(g(n))$, respectively)
    \item For any polynomial $p(n)$, no poly-time quantum algorithm can distinguish between the ensembles $\rho= \mathbb{E}_k \left[ \ket{\Psi_k}\bra{\Psi_k}^{\otimes p(n)}\right]$ and $\sigma= \mathbb{E}_k \left[ \ket{\Phi_k}\bra{\Phi_k}^{\otimes p(n)}\right]$ with more than negligible probability. That is, for any poly-time quantum algorithm $\mathcal{A}$, we have that
    \[ \left|\mathcal{A}(\rho) - \mathcal{A}(\sigma)\right| \leq \frac{1}{\text{negl}(n)}\]
\end{itemize}
Our definition is inspired by prior work of Gheorghiu and Hoban \cite{gheorghiu2020estimating}, who implicitly considered a similar notion. In our language, \cite{gheorghiu2020estimating} showed that PES ensembles exist with gap $n$ vs $n-k$ for any $k=O(1)$, based on LWE.  Our main result improves this construction to the maximum gap possible:

\begin{restatable}[High gap pseudoentangled states (informal)]{corollary}{Highentanglement}
\label{Thm:Highentanglement}
 There exists a pseudoentangled state ensemble (PES) with entanglement gap $\Theta(n)$ vs $\omega(\log n)$ across all cuts simultaneously, which is simultaneously a pseudorandom state ensemble, assuming there exists any quantum-secure OWF.
\end{restatable}

In contrast to Gheorghiu and Hoban's result, we achieve the maximum possible entanglement gap, are agnostic to the choice of quantum-secure OWF, applies to all cuts simultaneously, and simultaneously maintain indistinguishability from the Haar measure\footnote{One can indeed show \cite{gheorghiu2020estimating}'s construction is not itself a pseudorandom state ensemble, as we describe in Appendix~\ref{appendix:GHnotPRS}.}.
We similarly show our state can be instantiated in logarithmic depth.

\section{Main result}
\label{main construction}

In this section, we will prove our main result, Theorem \ref{Thm:Lowentanglement}.  To do this we first construct a pseudorandom quantum state with optimally low entropy across any cut. As discussed, this is the strongest possible notion of pseudoentanglement for a pseudorandom state ensemble, matching the lower bound established by Ji, Liu, and Song \cite{Ji2018}.

We then show how to \emph{tune} the entanglement entropy of our construction to achieve a pseudorandom state with entanglement entropy $\Theta(f(n))$ across each cut, for any function $f=\omega(\log(n))$.

\subsection{The subset phase state construction}
For particular choices of $S$ and a binary phase function $f:\{0,1 \}^n \to \{ 0, 1\}$, our pseudorandom state will have the following form.

\begin{align}\label{eqn:subset-phase}
\ket{\psi_{f, S}}=\frac{1}{\sqrt{|S|}}\sum_{x\in S}(-1)^{f(x)}\ket{x}.
\end{align}
Let us call states that are denoted by \eqn{subset-phase} ``subset phase states." Our next arguments are as follows.
\begin{itemize}
\item \textbf{Efficient preparation:} We show how to efficiently prepare subset phase states, when the subset and phases are chosen pseudorandomly, using appropriate quantum--secure pseudorandom functions and permutations.
\item \textbf{Proof of statistical closeness:} First, we will show that if $|S|= 2^{\omega(\log n)}$, and if $f$ is randomly chosen in \eqn{subset-phase}, then polynomially many copies of the corresponding density matrix are close in trace distance to polynomially many copies of a Haar random state. 

Qualitatively, this result means that a randomly chosen subset phase state is statistically close to a Haar random state, even with polynomially many copies.

\item \textbf{Proof of computational indistinguishability:} Then, we will show that conditioned on a cryptographic conjecture, we can efficiently prepare pseudorandom subset phase states that are computationally indistinguishable to a random subset phase state. 

The proof of security will hinge on a sequence of hybrids.

\item \textbf{Analysis of pseudoentanglement:} We will have a discussion on how our construction can be made to achieve the desired optimally low pseudoentanglement properties across any cut.

\item \textbf{Tight tunability:} Finally, we will discuss how to tightly tune the entanglement of our construction by varying the size of the subset.
\end{itemize}


\subsection{Notations}

We will use $\mathsf{TD}(\cdot~,~ \cdot)$ to denote the trace distance between two density matrices. Use $\perm_t$ to denote the set of all permutations among $t$ items. For any subset $S\subseteq\{0,1\}^n$ and any $\sigma\in \perm_t$, we define
\begin{align}
P_S(\sigma)\coloneqq\sum_{x_1,\ldots,x_t\in S}\ket{x_{\sigma^{-1}(1)},\ldots,x_{\sigma^{-1}(t)}}\bra{x_1,\ldots,x_t}.
\end{align}
Then,
\begin{align}\label{eqn:S-projector}
\Pisym^{S,t}=\frac{1}{t!}\sum_{\sigma\in\perm_t}P_S(\sigma)
\end{align}
is the projector onto the symmetric subspace of $(\C^{S})^{\otimes t}$. 

There are two sources of randomness in the subset phase state: randomness in choosing the subset and randomness in choosing the phase. To simplify notations, we use $\ket{\psi_S}$ to denote the subset phase state defined in Eq.~\eqn{subset-phase} with a random phase function and a fixed subset $S$, $\ket{\psi_f}$ to denote the subset phase state with a random subset and a fixed phase function, and $\ket{\psi}$ to denote the subset phase state where both the components are chosen at random. Sometimes, when the choice of the subset is specified by a function $p$, as we will see in the very next section, we slightly modify the notation in Eq.~\eqn{subset-phase} and use $|\psi_{f, p}\rangle$ to denote such an ensemble.

In the next section, we will show that we can efficiently instantiate these states using pseudorandom functions and permutations. We will reference this section later, in our final proof of the pseudorandom and pseudoentangled properties of these states.

\subsection{Efficiently preparing the state}
\label{efficient preparation}
Let $p$ be sampled uniformly at random from a family of quantum--secure pseudorandom permutations $P$ with 
\begin{equation}
P = \{p : [2^n] \rightarrow [2^n]\}
\end{equation}
and $f$ be sampled uniformly at random from a family of quantum--secure pseudorandom functions $F$ with
    \begin{equation}
F = \{f :[2^{n}] \to \{1, -1 \} \}.
\end{equation} Let us suppose we know $f$, $p$ and $p^{-1}$. We will give a recipe of how to efficiently prepare the state:
\begin{equation}
 |\psi_{f, p}\rangle = \frac{1}{\sqrt{2^k}} \sum_{x \in \{0, 1\}^k}  (-1)^{f(p(x 0^{\otimes (n-k)}))}\ket{p(x 0^{\otimes (n-k)})}
\end{equation}
The steps are as follows:

\begin{itemize}
    \item Start with $\ket{0^n}$.
    \item Let the size of the subset $\mathsf{S}$ be $2^k$ for some integer $k \leq n$.
    \item Apply $H^{\otimes k} \otimes I^{\otimes (n-k)}$ to $\ket{0^n}$.
    \item We get the state
    \begin{equation}
         \frac{1}{\sqrt{2^k}} \sum_{x \in \{0, 1\}^k} \ket{x 0^{\otimes (n-k)}}, 
    \end{equation} 
    where $x 0^{\otimes (n-k)}$ means we pad $n - k$ zeros to to end of $x$ to get an $n$-bit string.
    \item 
    
    We apply $p$ to this state to get
     \begin{equation}
          \frac{1}{\sqrt{2^k}} \sum_{x \in \{0, 1\}^k} \ket{x 0^{\otimes (n-k)}} \ket{p(x 0^{\otimes (n-k)})}.
    \end{equation}
    \item Finally, we apply the inverse of $p$, denoted by $p^{-1}$, to un-compute the first register. We get
     \begin{equation}
          \frac{1}{\sqrt{2^k}} \sum_{x \in \{0, 1\}^k}  \ket{p(x 0^{\otimes (n-k)})}.
    \end{equation}
    Observe that if $p$ were sampled from the set of truly random permutations, then this process would output a subset set state $\frac{1}{\sqrt{2^k}} \underset{{x\in S}}{\sum} \ket{x}$ uniformly at random from all subset states with size $2^k$.
    \item Finally, we construct a phase oracle using the description of $f$ to get the state
    \begin{equation}\label{eqn:pseudorandom-subset-phase}
          \frac{1}{\sqrt{2^k}} \sum_{x \in \{0, 1\}^k}  (-1)^{f(p(x 0^{\otimes (n-k)}))}\ket{p(x 0^{\otimes (n-k)})}.
    \end{equation}
    Note that if $P$ were a truly random permutation family and $f$ a truly random phase function, then the output distribution of this process is exactly the uniform distribution over subset phase states with $|\mathsf{S}| = 2^k.$
\end{itemize}


\subsection{Proof of statistical closeness to Haar random states}

In this section, we will prove that polynomially many copies of a subset phase state, where both the subset and the phases have been chosen at random, are statistically close to polynomially many copies of a Haar random state. More formally, we establish the following theorem.
\begin{theorem}\label{thmref:SPS-security}
For any $t<K\leq 2^n$, it holds that
\begin{equation}
    \mathsf{TD}\left(\underset{S\text{ with }|S|=K,\ f}{\E}\left[\ket{\psi_{f,S}}\bra{\psi_{f,S}}^{\otimes t}\right], \underset{\ket{\phi}\leftarrow\mathscr{H}(\C^N)}{\E}\left[\ket{\phi}\bra{\phi}^{\otimes t}\right]  \right)  < O\Bigg(\frac{t^2}{K}\Bigg),
    \end{equation}
where $\ket{\psi_{f,S}}$ is defined in \eqn{subset-phase}. That is to say,
\begin{equation}
    \mathsf{TD}\left(\underset{S\text{ with }|S|=K,\ f}{\E}\left[\ket{\psi_{f,S}}\bra{\psi_{f,S}}^{\otimes t}\right], \underset{\ket{\phi}\leftarrow\mathscr{H}(\C^N)}{\E}\left[\ket{\phi}\bra{\phi}^{\otimes t}\right]  \right)  < \frac{1}{\poly(n)},
    \end{equation}
    for $K = 2^{\omega(\log n)}$, any polynomially bounded $t$, and any $\poly(n)$, where $\mathscr{H}(\C^N)$ denotes the ensemble of Haar random states in the Hilbert space with dimension $N=2^n$.
\end{theorem}

\subsubsection{Useful results}

Before stating our proofs, let us state some useful results about Haar random states and symmetric subspaces. 
\begin{fact}\label{Haar-random}
\begin{align}
\underset{{\ket{\phi}\leftarrow\mathscr{H}(\C^N)}}{\E}\left[\ket{\phi}\bra{\phi}^{\otimes t}\right] &= \frac{\Pisym^{N,t}}{\Tr(\Pisym^{N,t})}.  
\end{align}
\end{fact}
\noindent Fact~\ref{Haar-random} states that polynomially many copies of a Haar random state can be interpreted as as a normalized projector onto the symmetric subspace of the full Hilbert space spanned by all $N$ basis states. This can be found in many places, including \cite{prs1}, and also in Fact 5.3 of \cite{ananth2023pseudorandom}.

\begin{fact}[Fact 5.4, \cite{ananth2023pseudorandom}]\label{fact:perp-trace-distance}
Let $\rho_1$, $\rho_2$ be density matrices such that $\rho_2=\alpha\rho_1+\beta\rho_1^{\perp}$ where $\rho_1\rho_1^{\perp}=\rho_1^{\perp}\rho_1=0$, $\alpha,\beta\in[0,1]$ and $\alpha+\beta=1$, then
\begin{align}
\TD(\rho_1,\rho_2)=\beta.
\end{align}
\end{fact}

Now, let us choose a subset $S$, then instead of a normalized projector onto the symmetric subspace of the full Hilbert space, defined a normalized projector onto the symmetric subspace of the truncated Hilbert space, spanned only by basis vectors in that subspace.

The following lemma states that if $S$ is chosen uniformly at random among all subsets with a fixed size, then on average, this new projector is ``close" to the projector onto the symmetric subspace of the full Hilbert space, in trace distance. The closeness depends on the size of the subset chosen.

\begin{lemma}\label{symmetric-subspace-expectation}
\begin{align}
\mathsf{TD}\left(\underset{S\text{ with }|S|=K}{\E}\left[\frac{\Pisym^{S,t}}{\Tr(\Pisym^{S,t})}\right],\frac{\Pisym^{N,t}}{\Tr(\Pisym^{N,t})}\right)\leq O(t^2/K),
\end{align}
where the expectation is taken over all subsets $S\subseteq\{0,1\}^n$ with fixed size $|S| = K$, and $K \geq t$.
\end{lemma}
Before proving Lemma~\ref{symmetric-subspace-expectation}, we first introduce the following useful concept from \cite{ananth2023pseudorandom}.

\begin{definition}[{\cite[Definition 5.6]{ananth2023pseudorandom}}]\label{defn:typeT}
Let $v\in[N]^t$ for some $N,t\in\mathbb{N}$, then define $\type(v)$ to be a vector in $[t+1]^N$ where the $i^{th}$ entry in $\type(v)$ denotes the frequency of $i$ in $v$. Further, let $T\in[t+1]^N$ for some $N,t\in\mathbb{N}$, then define
\begin{align}
\ket{\type_T}=\beta\sum_{v\in[N]^t\type(v)=T}\ket{v},
\end{align}
where $\beta\in\mathbb{R}$ is the normalization constant.
\end{definition}

The concepts $\type_v\in[t+1]^N$ and $\ket{\type_T}$ provide us a new way to depict the projectors onto the symmetric subspace as well as Haar random states. In particular, we have the following result.

\begin{lemma}\label{symmetry-type}
For all $N,t\in\mathbb{N}$, we have
\begin{align}\label{eqn:symmetry-type-1}
\TD\bigg(\mathop{\mathbb{E}}\limits_{T\leftarrow\{0,1\}^N, \mathrm{hamming}(T)=t}\ket{\type_T}\bra{\type_T},\frac{\Pisym^{N,t}}{\Tr(\Pisym^{N,t})}\bigg)\leq O\Bigg(\frac{t^2}{N}\Bigg),
\end{align}
where $\ket{\type_T}$ is defined in Definition~\ref{defn:typeT}. Furthermore, for all $N,n,t\in\mathbb{N}$ with $N=2^n$ and any $S\subseteq\{0,1\}^n$ with fixed size $|S| = K$ and $K \geq t$, we have
\begin{align}\label{eqn:symmetry-type-2}
\TD\bigg(\mathop{\mathbb{E}}\limits_{T\leftarrow\mathscr{T}(S,t)}\ket{\type_T}\bra{\type_T},\frac{\Pisym^{S,t}}{\Tr\left(\Pisym^{S,t}\right)}\bigg)\leq O\Bigg(\frac{t^2}{K}\Bigg),
\end{align}
where for any subset $S\subseteq\{0,1\}^2$ and any $t\in\mathbb{N}$,  $\mathscr{T}(S,t)$ denotes the set of vectors $T\in\{0,1\}^N$ with Hamming weight $t$ and the non-zero entries of $T$ have indices in the set $S$. Quantitatively,
\begin{align}\label{eqn:mathscrT-defn}
\mathscr{T}(S,t)\coloneqq\{T\in\{0,1\}^N\,|\,\mathrm{hamming}(T)=t\text{ and the non-zero entries of }T\text{ have indices in }S\}.
\end{align}
\end{lemma}
\begin{proof}
The proof of \eqn{symmetry-type-1} is given in Lemma 5.7 of \cite{ananth2023pseudorandom}, whereas the proof of \eqn{symmetry-type-2} closely follows the proof of Lemma 5.7 of \cite{ananth2023pseudorandom}. Specifically, by \eqn{S-projector} we denote
    \begin{align*}
        \rho\coloneqq\frac{\Pisym^{S,t}}{\Tr(\Pisym^{S,t})}
        &=\frac{1}{t!\Tr(\Pisym^{S,t})} \sum_{\sigma \in \perm_t} P_{S}(\sigma) \\
        &= \frac{1}{t!\Tr(\Pisym^{S,t})} \sum_{\sigma \in \perm_t} \sum_{x_1,\ldots,x_t\in S}\ketbra{x_{\sigma^{-1}(1)},\ldots,x_{\sigma^{-1}(t)}}{x_1,\ldots,x_t}.
    \end{align*}
    Moreover, we denote
    \begin{align}
        \sigma\coloneqq\mathop{\mathbb{E}}\limits_{T\leftarrow\mathscr{T}(S,t)} \ketbra{\type_T}{\type_T} 
        &= \mathop{\mathbb{E}}\limits_{T\leftarrow\mathscr{T}(S,t)} \left(\frac{1}{\sqrt{t!}}\sum_{\substack{v\in S^t\\ \type(v) = T}}\ket{v}\right)\left(\frac{1}{\sqrt{t!}}\sum_{\substack{v'\in S^t\\ \type(v') = T}}\bra{v'}\right) \\
        &= \frac{1}{t!} \mathop{\mathbb{E}}\limits_{T\leftarrow\mathscr{T}(S,t)} \left(\sum_{\substack{v,v'\in S^t\\ \type(v) = \type(v') = T}}\ketbra{v}{v'}\right) \\
        &= \frac{1}{t!} \mathop{\mathbb{E}}\limits_{T\leftarrow\mathscr{T}(S,t)} \left(\sum_{\substack{v\in S^t\\ \type(v) = T}}\sum_{\substack{\sigma\in \perm_t\\ v' = \sigma(v)}}\ketbra{v}{v'}\right) \\
        &= \frac{1}{t! {K \choose t}} \sum_{\substack{v\in S^t\\ \type(v) \in \{0,1\}^N }}\sum_{\sigma\in\perm_t}\ketbra{v}{\sigma(v)} \\
        &= \frac{1}{t! {K \choose t}} \sum_{\substack{x_1,\ldots,x_t\in S\\ x_1,\ldots,x_t \text{ are distinct}}}\sum_{\sigma\in \perm_t}\ketbra{x_1,\ldots,x_t}{x_{\sigma(1)},\ldots,x_{\sigma(t)}},
    \end{align}
    where the third line follows the fact that vector of same type as permutation of each other, and the fourth line is an expectation (since there are a total of $\binom{K}{t}$ strings of hamming weight of $t$ in $S$). The fifth line follows since $\type(v)\in S$, hence all elements of $v$ are distinct. Further, we define
    \begin{align}
    \sigma^{\perp} \coloneqq \frac{1}{t! \left({K + t - 1\choose t} - {K \choose t}\right)} \sum_{\substack{x_1,\ldots,x_t\in S\\ x_1,\ldots,x_t \text{ are not distinct}}}\sum_{\sigma\in \perm_t}\ketbra{x_1,\ldots,x_t}{x_{\sigma(1)},\ldots,x_{\sigma(t)}},
    \end{align}
    which satisfies $\sigma\sigma^{\perp}=\sigma^{\perp}\sigma=0$, and
    \begin{align}
    \rho=\alpha\sigma+\beta\sigma^{\perp},
    \end{align}
    where $\beta$ equals the probability of picking $x_1,\ldots,x_t\in S$ such that there is a collision, which is less than $O(t^2/K)$. Then by Fact~\ref{fact:perp-trace-distance}, we can derive that
    \begin{align}
    \TD\bigg(\mathop{\mathbb{E}}\limits_{T\leftarrow\mathscr{T}(S,t)}\ket{\type_T}\bra{\type_T},\frac{\Pisym^{S,t}}{\Tr\left(\Pisym^{S,t}\right)}\bigg)\leq O\Bigg(\frac{t^2}{K}\Bigg).
    \end{align}
    
\end{proof}

We are now ready to present the proof of Lemma~\ref{symmetric-subspace-expectation}.

\begin{proof}[Proof of Lemma~\ref{symmetric-subspace-expectation}]

By Lemma~\ref{symmetry-type}, we know that
\begin{align}
\TD\bigg(\mathop{\mathbb{E}}\limits_{T\leftarrow\mathscr{T}(S,t)}\ket{\type_T}\bra{\type_T},\frac{\Pisym^{S,t}}{\Tr(\Pisym^{S,t})}\bigg)\leq O\Bigg(\frac{t^2}{K}\Bigg)
\end{align}
for any $S$ with size $K$, and
\begin{align}
\TD\bigg(\mathop{\mathbb{E}}\limits_{T\leftarrow\{0,1\}^N, \mathrm{hamming}(T)=t}\ket{\type_T}\bra{\type_T},\frac{\Pisym^{N,t}}{\Tr(\Pisym^{N,t})}\bigg)\leq O\Bigg(\frac{t^2}{N}\Bigg).
\end{align}
For any $T\in\{0,1\}^N$ with hamming$(T)=t$, denote
\begin{align}
p_T=\Pr_{T'\leftarrow\mathscr{T}(S,t)}\{T=T'\},\quad S\text{ is a random subset of }\{0,1\}^n\text{ with size }|S|,
\end{align}
where $\mathscr{T}(S,t)$ defined in Eq.~\eqn{mathscrT-defn} denotes the set of vectors $T\in\{0,1\}^N$ with Hamming weight $t$ and the non-zero entries of $T$ have indices in the set $S$. Observe that $p_T$ is a uniform distribution among all possible $T\in\{0,1\}^N$ with hamming$(T)=t$, which leads to
\begin{align}
\underset{S\text{ with }|S|=K}{\E}\Bigg[\mathop{\mathbb{E}}\limits_{T\leftarrow\mathscr{T}(S,t)}\ket{\type_T}\bra{\type_T}\Bigg]
=\mathop{\mathbb{E}}\limits_{T\leftarrow\{0,1\}^N, \mathrm{hamming}(T)=t}\ket{\type_T}\bra{\type_T}.
\end{align}
Hence,
\begin{align}
&\TD\bigg(\underset{S\text{ with }|S|=K}{\E}\Big[\frac{\Pisym^{S,t}}{\Tr(\Pisym^{S,t})}\Big],\frac{\Pisym^{N,t}}{\Tr(\Pisym^{N,t})}\bigg)\nonumber\\
&\qquad\leq \TD\bigg(\mathop{\mathbb{E}}\limits_{T\leftarrow\{0,1\}^N, \mathrm{hamming}(T)=t}\ket{\type_T}\bra{\type_T},\frac{\Pisym^{N,t}}{\Tr(\Pisym^{N,t})}\bigg)\nonumber\\
&\qquad\quad\ +\underset{S\text{ with }|S|=K}{\E}\bigg[\TD\bigg(\mathop{\mathbb{E}}\limits_{T\leftarrow\mathscr{T}(S,t)}\ket{\type_T}\bra{\type_T},\frac{\Pisym^{S,t}}{\Tr(\Pisym^{S,t})}\bigg)\bigg]\\
&\qquad\leq O(t^2/K).
\end{align}
\end{proof}

\subsubsection{Proof of Theorem~\ref{thmref:SPS-security}}

Our proof of Theorem $3.1$ will go through two main steps.
\begin{itemize}
\item First, we consider a fixed subset and a random phase function. Then we prove that 
\begin{lemma}
\label{lemma: fixed subset random phase}
\begin{align}
\mathsf{TD}\left(\underset{f}{\E}\left[\ket{\psi_{f,S}}\bra{\psi_{f,S}}^{\otimes t}\right],\frac{\Pisym^{S,t}}{\Tr(\Pisym^{S,t})}\right)\leq \mathcal{O}(t^2/|S|).
\end{align}
\end{lemma}
\begin{proof}
The proof closely follows the proof of Theorem~5.1 in \cite{ananth2023pseudorandom} which also proceeds via a hybrid argument. In particular, we consider the following hybrids.
\paragraph{Hybrid 1.} Choose a binary phase function $f:\{0,1 \}^n \to \{ 0, 1\}$ uniformly at random. Output $t$ copies of the subset phase state $\ket{\psi_{f,S}}$ defined in \eqn{subset-phase}. Note that the output density matrix equals
\begin{align}
\underset{f}{\E}\left[\ket{\psi_{f,S}}\bra{\psi_{f,S}}^{\otimes t}\right]
\end{align}
In expectation.

\paragraph{Hybrid 2.} Sample $w \in S^t$ uniformly at random under the condition that the non-zero entries of $w$ have indices in the set $S$. Let $T = \type(w) \pmod{2}$.
Output $\ket{\bintype_T}^{\otimes t}$, where
\begin{align}
\ket{\bintype_T}\coloneqq \beta\sum_{\substack{v\in S^t\\ \type(v) \pmod{2} = T}}\ket{v}.
\end{align}

\paragraph{Hybrid 3.} Sample $T \in \mathscr{T}(S,t)$ uniformly at random, where $\mathscr{T}(S,t)$ is defined in \eqn{mathscrT-defn}. Output $\ket{\type_T}^{\otimes t}$.

We first show that \textbf{Hybrid 1} and \textbf{Hybrid 2} are identical. The output density matrix $\rho$ of \textbf{Hybrid 1} satisfies
\begin{align}
\rho&=\underset{f}{\E}\left(\frac{1}{\sqrt{|S|^t}}\sum_{x_1,\ldots,x_t\in S}(-1)^{f(x_1)+\cdots+f(x_t)}\ket{x_1,\ldots,x_t}\right)\nonumber\\
&\qquad\qquad\quad\times \left(\frac{1}{\sqrt{|S|^t}}\sum_{y_1,\ldots,y_t\in S}(-1)^{f(y_1)+\cdots+f(y_t)}\bra{y_1,\ldots,y_t}\right)\\
&=\frac{1}{|S|^t}\underset{f}{\E}\left(\sum_{\substack{x_1,\ldots,x_t\in S\\ y_1,\ldots,y_t\in S}}(-1)^{f(x_1)+\cdots+f(x_t)+f(y_1)+\cdots+f(y_t)}\ket{x_1,\ldots,x_t}\bra{y_1,\ldots,y_t}\right)\\
&=\frac{1}{|S|^t}\left(\sum_{\substack{x_1,\ldots,x_t\in S\\ y_1,\ldots,y_t\in S}}\sum_{\substack{\type(x)\text{ mod }2=\type(y)\text{ mod }2}}\ket{x_1,\ldots,x_t}\bra{y_1,\ldots,y_t}\right),\label{eqn:rho-final}
\end{align}
whereas the output density matrix $\rho'$ of \textbf{Hybrid 2} satisfies
\begin{align}
    \rho' &= C_1\cdot\underset{w\in S^t}{\E} \left(\sum_{\substack{v\in S^t \\ \type(v)\bmod{2} = \type(w)\bmod{2}}}\ket{v}\right)\left(\sum_{\substack{v'\in S^t \\ \type(v')\bmod{2} = \type(w)\bmod{2}}}\bra{v'}\right) \\
    &= C_2\cdot\frac{1}{|S|^{t}} \left(\sum_{\substack{v,v'\in S^t \\ \type(v)\bmod{2} = \type(v')\bmod{2}}}\ketbra{v}{v'}\right),
\end{align}
where $C_1$, $C_2$ are appropriate normalization factors. Note that it is the same as $\rho$ as shown in \eqn{rho-final}, indicating that the output of \textbf{Hybrid 1} and \textbf{Hybrid 2} are identical to each other.

Let $w\in S^t$ be sampled uniformly at random and $T = \type(w) \pmod{2}$, then $\mathrm{hamming}(T) = t$ with probability $1-t^2/|S|$ as $w$ has no collisions with probability $1-t^2/|S|$. Hence, the trace distance between the outputs of \textbf{Hybrid 2} and \textbf{Hybrid 3} is $O(t^2/|S|)$. By Lemma~\ref{symmetry-type}, the trace distance between the output of \textbf{Hybrid 3} and 
\begin{align}
\frac{\Pisym^{S,t}}{\Tr(\Pisym^{S,t})}
\end{align}
is also $O(t^2/|S|)$, indicating
\begin{align}
\mathsf{TD}\left(\underset{f}{\E}\left[\ket{\psi_{f,S}}\bra{\psi_{f,S}}^{\otimes t}\right],\frac{\Pisym^{S,t}}{\Tr(\Pisym^{S,t})}\right)\leq \mathcal{O}(t^2/|S|),
\end{align}
since $\underset{f}{\E}\left[\ket{\psi_{f,S}}\bra{\psi_{f,S}}^{\otimes t}\right]$ is the output of \textbf{Hybrid 1}.
\end{proof}

Lemma \ref{lemma: fixed subset random phase} means that for appropriately large choices of $S$, if we fix the subset but pick the binary phases at random, then the corresponding density matrix is \emph{close in trace distance} to the normalized symmetric projector onto the truncated Hilbert space defined by that fixed subspace. The closeness depends on the size of the subset chosen.
\item Then using Fact \ref{Haar-random} in the second line, the triangle inequality for trace distance along with Lemma \ref{symmetric-subspace-expectation} in the third line, a simple convexity argument in the fourth line, and Lemma \ref{lemma: fixed subset random phase} in the fifth line, we write
\begin{equation}
\begin{aligned}
    &\mathsf{TD}\left(\underset{S\text{ with }|S|=K,\ f}{\E}\left[\ket{\psi}\bra{\psi}^{\otimes t}\right], \underset{\ket{\phi}\leftarrow\mathscr{H}(\C^N)}{\E}\left[\ket{\phi}\bra{\phi}^{\otimes t}\right]  \right) \\ 
    &= \mathsf{TD}\left(\underset{S\text{ with }|S|=K,\ f}{\E}\left[\ket{\psi}\bra{\psi}^{\otimes t}\right], \frac{\Pisym^{N,t}}{\Tr(\Pisym^{N,t})}  \right) \\
    &\leq \mathsf{TD}\left(\underset{S\text{ with }|S|=K,\ f}{\E}\left[\ket{\psi}\bra{\psi}^{\otimes t}\right], \underset{S~\text{with}~|S|=K}{\E}\left[\frac{\Pisym^{S,t}}{\Tr(\Pisym^{S,t})}\right]  \right) + \mathcal{O}(t^2/K) \\
        &\leq \mathsf{TD}\left(\underset{f}{\E}\left[\ket{\psi}\bra{\psi}^{\otimes t}\right], \frac{\Pisym^{S,t}}{\Tr(\Pisym^{S,t})}  \right) + \mathcal{O}(t^2/K),~~~~~\forall S\text{ s.t. }|S|=K  \\
        &\leq \mathcal{O}(2t^2/K),
    \end{aligned}
\end{equation}

Hence, as long as we choose $K$ to be quasi-polynomially large, the trace distance between the subset phase state $\ket{\psi}$ defined in Eq.~\eqn{subset-phase} is inverse quasi-polynomially close to Haar random state in trace distance, given that $t=\poly(n)$.
\end{itemize}

\subsection{Proof of computational indistinguishability}

\noindent In this section, we will prove the following theorem.
\begin{theorem}
Consider an ensemble of subset phase states $|\psi_{A, p}\rangle$ given by
\begin{equation}
\label{main state}
     |\psi_{f, p}\rangle = \frac{1}{\sqrt{2^k}} \sum_{x \in \{0, 1\}^k}  (-1)^{f(p(x 0^{\otimes (n-k)}))}\ket{p(x 0^{\otimes (n-k)})},
\end{equation}
where $p$ is sampled uniformly at random from a family of quantum--secure pseudorandom permutations $P$ with 
\begin{equation}
P = \{p : [2^n] \rightarrow [2^n]\}
\end{equation}
and $f$ is sampled uniformly at random from a family of quantum--secure pseudorandom functions $F$ with
    \begin{equation}
F = \{f :[2^{n}] \to \{1, -1 \} \}.
\end{equation}
Then, \eqref{main state} defines an ensemble of pseudorandom quantum states, with the secret key $\mathsf{K}$ being the description of $f$, $p$, and $p^{-1}$, where $p^{-1}$ is the inverse permutation of $p$. 
\end{theorem}

Note that the efficient preparability of the states in \eqref{main state} follow from Section \ref{efficient preparation}. The rest of the proof will be a security analysis: to prove that this construction is computationally indistinguishable from Haar random states. Finally, in a separate section, we will analyze the entropy of this state ensemble.


\subsubsection{Security analysis}
In this subsection, we prove the following proposition.

\begin{proposition}
The ensemble of subset phase states defined in \eqn{pseudorandom-subset-phase} is computationally indistinguishable from Haar random states, with the secret key being the description of $f$ and $p$, when $|\mathsf{S}| = 2^{\omega(\log n)}$.
\end{proposition}

\begin{proof}
Note that $f$ is oracle indistinguishable from a random function $r_\mathsf{f}$, from the security analysis in Section \ref{oracle indistinguishability}. Additionally, by definition, $p$ is oracle indistinguishable from a truly random permutation $r_\mathsf{p}$. Moreover, again by definition, the oracle indistinguishability result holds even when the adversary is given access to the inverse of the permutation. That is, no adversary can distinguish between $(p, p^{-1})$ and $(r_\mathsf{p}, r_\mathsf{p}^{-1})$ when given black box access and promised one of these is the case.

So, when given access to three black boxes, promised to either $(f, p, p^{-1})$, or $(r_\mathsf{f}, r_\mathsf{p}, r_\mathsf{p}^{-1})$, no polynomial time adversary, with query access, can distinguish between these two cases. Now, the following sequence of hybrids completes the proof.

\paragraph{Hybrid 0.} This is the case where the adversary is given polynomially many copies of the state
\begin{equation}
\label{first state}
    |\psi_{f, p}\rangle =  \frac{1}{\sqrt{2^k}} \sum_{x \in \{0, 1\}^k}  (-1)^{f(p(x 0^{\otimes (n-k)}))}\ket{p(x 0^{\otimes (n-k)})}.
\end{equation}

\paragraph{Hybrid 1.} This is the case where the adversary is given polynomially many copies of the state
\begin{equation}
\label{second state}
    |\psi_{R, r}\rangle =  \frac{1}{\sqrt{2^k}} \sum_{x \in \{0, 1\}^k}  (-1)^{r_{\mathsf{f}}(r_{\mathsf{p}}(x 0^{\otimes (n-k)}))}\ket{r_{\mathsf{p}}(x 0^{\otimes (n-k)})}.
\end{equation}
This is computationally indistinguishable from $\textbf{Hybrid 0}$ because, otherwise, we can efficiently distinguish between $(f, p, p^{-1})$ and  $(r_\mathsf{f}, r_\mathsf{p}, r_\mathsf{p}^{-1})$ when given black box access, by using the unknown black box to prepare polynomially many copies of a state that has to be either \eqref{first state} or \eqref{second state}. 

\paragraph{Hybrid 2.} The adversary is given polynomially many copies of a Haar random state. This is indistinguishable from \textbf{Hybrid 1}  from Theorem \ref{thmref:SPS-security}. 
\end{proof}

\subsection{Entanglement entropy of pseudorandom subset phase states}
Let $|\psi_{A, p}\rangle$ be a pseudorandom subset phase state and let
\begin{equation}
    \rho_{A, p} = |\psi_{A, p}\rangle \langle \psi_{A, p}|.
\end{equation}
To prevent cluttering notations, we will drop $A$ and $p$ from the subscript of $\rho$ and take them to be implicit whenever we use the symbol, unless otherwise stated. Let ${S}$ be the size of the subset defined by $p$ and let $\mathsf{S}(\cdot)$ be the von Neumann entanglement entropy of a density matrix. 

For an $n$--qubit state $\ket{\psi}$, let $(\mathsf{X}, \mathsf{Y})$ be any partition of the $n$ qubits. Then, for reduced density matrices $\rho_{\mathsf{X}}$ and $\rho_{\mathsf{Y}}$, let the von Neumann entropy, for each, be denoted by $\mathsf{S}(\rho_{\mathsf{X}:\mathsf{Y}})$. 

The following statements are immediate.
\begin{corollary}
\label{first corollary}
For \emph{any} cut $(\mathsf{X}, \mathsf{Y})$ of $n$ qubits, such that $|\mathsf{X}| + |\mathsf{Y}| = n$,
\begin{equation}
    \mathsf{S}(\rho_{\mathsf{X}:\mathsf{Y}}) = \mathcal{O}(\log| {S} |).
\end{equation}
\end{corollary}
\begin{proof}
    The proof follows trivially from noting that the rank of the density matrix $\rho$ is at most $|{S}|$, as it is a density matrix corresponding subset state over a subset of size $|{S}|$ and has, at most $|{S}|$ linearly independent rows or columns.
\end{proof}

\begin{corollary}
For $|{S}| = 2^{\poly\log n}$, for \emph{any} cut $(\mathsf{X}, \mathsf{Y})$ of $n$ qubits, such that $|\mathsf{X}| + |\mathsf{Y}| = n$,
\begin{equation}
    \mathsf{S}(\rho_{\mathsf{X}:\mathsf{Y}}) = \Theta(\poly\log n).
\end{equation}
\end{corollary}
\begin{proof}
    The upper bound follows from Corollary \ref{first corollary} and the lower bound follows from the SWAP test, as described in Appendix \ref{SWAPtest}.
\end{proof}

This shows that when $|{S}| = 2^{\text{poly}\log n}$, we get optimally low pseudoentanglement across every cut, no matter what the spatial geometry is. We will now see a way of tuning the entanglement entropy by varying the size of the subset.

\subsection{Tuning the entanglement entropy of the random subset phase state construction}

By varying the size of the subset, we can tune the entanglement entropy of our random subset phase state construction. However, since the SWAP test lower bound of $\Omega(\log n)$ is no longer tight for these cases, we need a different way of proving a tight lower bound. For that, we will consider a very specific form of the pseudorandom phase function. 

This is what we will motivate and discuss in the next parts. Before that discussion, just for convenience of analysis, we will define a pseudorandom matrix.

\subsubsection{Pseudorandom matrices for subset phase states}
Let a given subset be ${S}$, a subset state $|\psi_{f, S}\rangle$, and a given partition $(\mathsf{X}, \mathsf{Y})$, where $|\mathsf{X}| = m$ and $| \mathsf{Y}| = n - m$. Let us write the reduced density matrix across the partition $\mathsf{X}$. Let $|{S}| = 2^k$. 

\begin{align}
\rho_{\mathsf{X}} 
&= \frac1{2^k} \left(\sum_{i \in \{ 0, 1\}^{m}, j \in \{ 0, 1\}^{n-m}, ij \in {S} }~   \sum_{k \in \{ 0, 1\}^{m}, l \in \{ 0, 1\}^{n-m}, kl \in {S} } (-1)^{f(i, j)+f(k, l)} \Tr_2( \ket{i} \ket{j} \bra{k}\bra{l} )\right) \\
& = \frac1{2^k} \left(\sum_{i, k \in \{ 0, 1\}^{m}, j \in \{ 0, 1\}^{n-m}, ij \in {S}, kj \in {S} }   (-1)^{f(i, j)+f(k, j)}  \ket{i} \bra{k} \right) \\
\label{dropB1}
& = \frac1{2^k} \left(\sum_{i, k \in \{ 0, 1\}^{m}, j \in \{ 0, 1\}^{n-m}, ij \in {S}, kj \in {S}} B^{}_{i, j}B^{}_{k, j}\ket{i} \bra{k} \right) \\
\label{dropB2}
& = \frac1{2^k} \left(\sum_{i, k \in \{ 0, 1\}^{m}, j \in \{ 0, 1\}^{n-m}, ij \in {S}, kj \in {S} } B^{}_{i, j}B^{}_{k, j}\ket{i}\langle j| j \rangle \bra{k} \right) \\
\label{dropB3}
& = \frac1{2^k} \left(\sum_{i \in \{ 0, 1\}^{m} }  \sum_{j \in \{ 0, 1\}^{n-m}, ij \in {S} } B^{}_{i, j}\ket{i}\langle j| \right) \left(\sum_{j \in \{ 0, 1\}^{n-m} }  \sum_{k \in \{ 0, 1\}^{m}, kj \in {S} } B^{}_{k, j}\ket{j}\langle k| \right) \\
\label{section2.1-last-line}
&=\frac{1}{2^k} BB^{\mathsf{T}},
\end{align}
where we define a pseudorandom matrix $B_{\mathsf{X}:\mathsf{Y}, f}$ as follows.
\begin{equation}
\begin{aligned}
    B_{\mathsf{X}:\mathsf{Y}, f, i, j} &= f(i, j)~~~~~\text{when}~~ij \in {S}, ~i \in \{0, 1\}^m,~ j \in \{0, 1\}^{n-m} \\
    &=0~~~~~~~~~~~~\text{otherwise}.
    \end{aligned}
\end{equation}
When the context is clear, we drop the corresponding subscripts from our notation, which we have done in \eqref{dropB1}, \eqref{dropB2}, \eqref{dropB3}, and \eqref{section2.1-last-line}. 
\subsubsection{Tuning the entanglement entropy}
For tuning the entanglement entropy, we will consider  an ensemble of pseudorandom subset phase states $|\psi_{f, p}\rangle$, where $f$ is chosen as
\begin{equation}
\label{choosingf}
f(i):= h(q(i)),
\end{equation}
where $h$ is uniformly drawn from a $4$-wise independent function family 
\begin{equation}
H = \{h :[2^{n}] \to \{1, -1 \} \},
\end{equation}
and $q$ is uniformly drawn from $Q$ --- a quantum-secure pseudorandom permutation (PRP) family --- where
\begin{equation}
    Q = \{q :[2^{n}] \to [2^{n}] \}.
\end{equation}
Note that by a simple hybrid argument, $f$ is computationally indistinguishable from a truly random function: so, it is both pseudorandom and $4$-wise independent. This is proven in detail in the Appendix, in Section \ref{oracle indistinguishability}, \ref{security}, and \ref{security2}. We will now prove the following theorem.
\begin{theorem}
    Let $\omega(\log n) \leq k \leq n$ and let $|{S}|=2^k$. Consider a cut  $(\mathsf{X}, \mathsf{Y})$ of $n$ qubits, such that $|\mathsf{X}| + |\mathsf{Y}| = n$ and $|\mathsf{X}|, |\mathsf{Y}| \geq k$. Let the pseudorandom phase function satisfy \eqref{choosingf}. Then, with high probability over the choice of the state, 
    \begin{equation}
        \mathsf{S}(\rho_{\mathsf{X}:\mathsf{Y}}) = \Theta(k).
    \end{equation}
\end{theorem}
\begin{proof}
    The upper bound follows from Corollary \ref{first corollary}. For the lower bound, we will use the inequality,
    \begin{equation}
    \label{second inequality}
    \mathsf{S}(\rho_{\mathsf{X}:\mathsf{Y}}) \geq -\log \left(\bigg|\bigg|\frac{1}{2^k} B_{\mathsf{X}:\mathsf{Y}}{B_{\mathsf{X}:\mathsf{Y}}}^{\mathsf{T}} \bigg|\bigg|_F\right),
    \end{equation}
    just as we did for Section \ref{subsection: high entropy}, where $B_{\mathsf{X}:\mathsf{Y}}$ is the pseudorandom matrix corresponding to the partition $(\mathsf{X}, \mathsf{Y})$. \eqref{second inequality} can be derived in the same way as \eqref{rank and entropy}, by  Jensen's inequality. Hence, it suffices to lower bound the quantity
    \begin{equation}
        \log \left(\bigg|\bigg|\frac{1}{2^k} B_{\mathsf{X}:\mathsf{Y}}{B_{\mathsf{X}:\mathsf{Y}}}^{\mathsf{T}} \bigg|\bigg|_F\right).
    \end{equation}
    In this proof, for simplicity, we prove the statement for partitions of size $n/2$. Note that the same proof follows for any other partition, just by changing the dimensions of the matrix $B_{\mathsf{X}:\mathsf{Y}}$.
    
    Having fixed the partition, let us drop the subscripts from $B$, to avoid any redundant notational clutter. Note that, 
	\begin{align*}
		& \mathrm{E} \left [ \left \|\frac{1}{2^{k}} BB^{\mathsf{T}} \right \|^2_F \right] \\ &= \frac{1}{2^{2k}} \mathrm{E} \left [ \left \| BB^{\mathsf{T}} \right \|^{2}_F \right] \\
		& = \frac{1}{2^{2k}} \sum_{i=1}^{2^{k/2}}  \sum_{j=1}^{2^{n/2}} \mathrm{E} \left[\left ( \sum_{l=1}^{2^{n/2}}B_{il}\cdot B_{jl}\right)^2 \right] \\
		& = \frac{1}{2^{2k}} \sum_{i=1} ^{2^{n/2}}  \mathrm{E} \left[\left ( \sum_{l=1}^{2^{n/2}}B_{il}\cdot B_{il} \right)^2 \right] + \frac{1}{2^{2k}} \sum_{i \ne j, i, j = 1}^{2^{n/2}} \mathrm{E} \left[\left ( \sum_{l=1}^{2^{n/2}}B_{il}\cdot B_{jl}\right)^2 \right] \\
		& = \frac{1}{2^{2k}}  \sum_{i=1} ^{2^{n/2}}  \mathrm{E} \left[\left ( \sum_{l=1}^{2^{n/2}} B_{il} \right ) +  2 \left(\sum_{l \ne l', l, l' = 1}^{2^{n/2}} B_{il}\cdot B_{il'} \right) \right] +  \\
		& \quad \quad \frac{1}{2^{2k}} \sum_{i \ne j, i, j = 1}^{2^{n/2}} \mathrm{E} \left[\left ( \sum_{l=1}^{2^{n/2}}B_{il}\cdot B_{jl}\right)^2 \right] \\
		& \le \frac{1}{2^{2k}} \left ( 2^k +2^{n/2 + 1} \cdot 2^n \cdot \left ( \frac{2^k}{2^n} \right)^2 \right) + \frac{1}{2^{2k}} \sum_{i \ne j, i, j = 1}^{2^{n/2}} \sum_{l=1}^{2^{n/2}}\mathrm{E} \left[\left ( B_{il}\cdot B_{jl}\right)^2 \right] \\
		& \le \frac{1}{2^{k - 1}} + \frac{2^n}{2^{2k}} 2^{n/2} \frac{2^{2k}}{2^{2n}} \\
		& \le \frac{1}{2^{k/2 - 1}},
	\end{align*}
    where we have used the fact that because $f$ is $4$--wise independent, conditioned on any choice of $\mathsf{S}$ we have
    \begin{align*}
        \mathrm{E} \left[\left ( \sum_{l=1}^{2^{n/2}}B_{il}\cdot B_{jl}\right)^2 \right] & = \sum_{l=1}^{2^{n/2}}\mathrm{E} \left[\left ( B_{il}\cdot B_{jl}\right)^2 \right] \\
        \le 2^{n/2} \frac{2^{2k}}{2^{2n}}.
    \end{align*}
Finally, by the Markov's inequality, we have 
	\begin{equation}
	\Pr \left [ \left \|\frac{1}{2^{k}} BB^{\mathsf{T}} \right \|^2_F  > 2^{-k/4}\right] \le 2^{1 - k/2}.
	\end{equation}
	Therefore, \begin{equation}
	\Pr \left [ \left \|\frac{1}{2^{k}} BB^{\mathsf{T}} \right \|_F  > 2^{-k/8}\right] \le 2^{1 - k/2}.
	\end{equation}
 Hence, the proof follows.
\end{proof}

\subsection{Instantiating our constructions using low depth circuits}
\label{construction}
A natural question to ask is how we can explicitly  construct our two pseudorandom states. Note that pseudorandom functions, $4$-wise independent functions, and pseudorandom permutations can be instantiated using one-way functions \cite{WEGMAN1981265, Katz2014}. So, our pseudorandom states can also be instantiated using quantum-secure one-way functions. 

Moreover, we can instantiate our states using low-depth circuits. There are three components in our constructions, which we will instantiate one by one. 

\begin{itemize}
    \item Pseudorandom functions: Pseudorandom functions can be implemented in the complexity class $\mathsf{NC}^{1}$ (the class of $\mathcal{O}(\log n)$ depth circuits with bounded fan-in but potentially unbounded fan-out), by using the Naor-Reingold construction \cite{Naor1999}. The construction is secure assuming the existence of pseudorandom synthesizers, which can again be constructed in $\mathsf{NC}^{1}$ assuming the existence of primitives like weak pseudorandom functions or trapdoor one-way permutations. From \cite{Zha15}, the quantum security of this construction depends on the quantum security of the underlying primitives. So, from \cite{Naor1999} and \cite{Zha15}, quantum-secure Naor-Reingold pseudorandom functions exist assuming the existence of quantum-secure trapdoor one-way permutations or quantum-secure weak pseudorandom functions.
    
    There are other efficient short depth constructions of quantum-secure pseudorandom functions. For instance, the pseudorandom function given in \cite{BPR}, based on the hardness of Learning With Errors (LWE), can be compiled in $\mathsf{NC}^{1}$. In \cite{Zha15}, Zhandry gives an explicit proof of quantum security of this function.  This is another function we could potentially use, instead of the Naor-Reingold based construction.
    \item Pseudorandom permutations: We use the $4$-round Luby-Rackoff construction of pseudorandom permutations \cite{Luby1988}. This was proven to be quantum-secure in \cite{Luby-Rackoff_secure}. For the round function in the Luby-Rackoff construction, we use one of the $\mathsf{NC}^{1}$ implementable pseudorandom functions from either \cite{Naor1999} or \cite{BPR_2012}. Hence, the security of this construction can also be argued from the security of the primitives underlying the pseudorandom function we use in its construction. Since there are only four rounds, the most depth-intensive part of the Luby-Rackoff construction is the pseudorandom round-function for each round: consequently, the Luby-Rackoff construction can be implementable in $\mathsf{NC}^{1}$ using a pseudorandom function that can be implementable in $\mathsf{NC}^{1}$.
    
    \item $4$-wise independent function: We could use a $4$-wise independent hash function, of appropriate domain and range, of the type in \cite{WEGMAN1981265}, based on polynomials. Since addition and multiplication modulo a prime can be done in $\mathsf{NC}^{1}$, the function is implementable in $\mathsf{NC}^{1}$. 
    
    We could also instantiate 4-wise independent functions using BCH codes \cite{Bose1960}, which are implementable in $\mathsf{NC}^{1}$, as matrix multiplication and other standard linear algebra techniques are in $\mathsf{NC}^{1}$. 
\end{itemize}

\noindent Note that to compile these primitives by log-depth quantum circuits, the fact that these primitives are computable by  $\mathsf{NC}^{1}$ circuits is not enough. This is because any quantum gate set usually has bounded fan-out gates, and compiling log-depth circuits with potentially unbounded fan-out with that gate set can blow up the depth to linear. 

In more formal terms, it is not clear whether these primitives can be compiled by $\mathsf{QNC}^{1}$ circuits -- these are the quantum analogues of $\mathsf{NC}^{1}$ circuits, with the only difference being that they do not have unbounded fan-out gates. However, they can be implemented in $\mathsf{QNC}^{1}_\mathsf{f}$: these are $\mathcal{O}(\log n)$ depth quantum circuits with bounded fan-in gates and special ``quantum fan-out" gates. Hence, our pseudorandom states are implementable in $\mathsf{QNC}^{1}_\mathsf{f}$.


\section{Applications}
In this section, we will describe the applications of our construction. 

\subsection{Low entropy pseudorandom states imply inefficient entropy distillation protocols}
\label{app:distillation}
In this section, we will discuss connections between our pseudorandom state constructions and entanglement distillation.  

Consider $m$ copies of an uknown $d$-dimensional quantum state $\ket{\psi}$. Consider a bipartition $(A, B)$ of the qubits in $\ket{\psi}$. Let $\rho^{A}$ and $\rho^{B}$ be the reduced density matrix across each bipartition, and, to avoid clutter of notation, let $\mathsf{S}(\rho) = \mathsf{S}(\rho^{A}) = \mathsf{S}(\rho^{B})$ be the von Neumann entropy across each bipartition. Then we know, due to previous results:
\begin{lemma}\cite{hayashi, harrow}
\label{protocol1}
Given an unknown $\ket{\psi}^{\otimes m}$, there is an $\mathsf{LOCC}$ protocol, which runs in $\text{poly}(n)$ time, to get at least $p$ EPR pairs, where
\begin{equation}
\label{distillation}
    p \geq m \left(\mathsf{S}(\rho) - \eta(\delta) - \delta \log d \right) - \frac{1}{2} d (d + 1) \log (m + d),
\end{equation}
with probability at least
\begin{equation}
    1 - \mathsf{exp}\left( \frac{-n \delta^2}{2} \right)\left(n + d \right)^{d (d+1)/2},
\end{equation}
where $\eta(\cdot)$ is the binary entropy function.
\end{lemma}
The protocol involves applying a Schur transform to $\ket{\psi}^{\otimes m}$ and then measuring in the standard basis. Note that the Schur transform can be efficiently implemented in $\text{poly}(n, \log d)$ time, using \cite{Krovi_2019}, up to inverse exponential precision. 
A trivial upper bound to $p$ is $\mathsf{S}(\rho)$ --- one cannot distill more EPR pairs than the amount of distillable entanglement entropy present, which, for pure states, is equal to the von Neumann entropy $\mathsf{S}(\rho)$ \cite{horodecki}. However, distilling all of the distillable entanglement entropy is non-trivial and the upper bound could potentially be very loose.

Note that when $d = 2^n$, and $m = \text{poly}(n)$, the RHS in \eqref{distillation} is negative, for any value of $\mathsf{S}(\rho)$. Hence, the lower bound on $p$ is vacuous as $p \geq 0$. Tighter lower bounds to $p$ are not known. So, the distillation protocol could essentially terminate without generating a single EPR pair.

We will sketch an argument that for any efficient distillation protocol, working with polynomially many copies of $\ket{\psi}$, the lower bounds on $p$ are unlikely to be too tight. At a high level, our sketch would show that assuming a cryptographic conjecture, no efficient entanglement distillation protocol, working with polynomially many copies of an unknown quantum state, can guarantee a distillation of more than polylogarithmically many EPR pairs.

\begin{proposition}
For an unknown quantum state $\ket{\psi}^{\otimes m}$ with $m = \poly(n)$, $d = 2^n$, and for a bipartition where each side has $\Omega(n)$ qubits, there is no efficient distillation protocol such that the number of EPR pairs produced $p = \omega(\poly\log \mathsf{S}(\rho))$ with non-negligible probability, assuming the existence of quantum secure one-way functions.
\end{proposition}
\begin{proof}
Assume the contrapositive. Choose a bipartition of size $n/2$ \footnote{A similar argument works for any bipartition where each side has size $\Omega(n)$.}. For a Haar random state, $\mathsf{S}(\rho) = \Theta(n)$ \cite{Page_1993}. 
Consequently, we can distill between  $\omega(\text{poly}\log n)$ to $\mathcal{O}(n)$ EPR pairs from this state. However, we can distill between $\omega(\log(\log n))$  to $\mathcal{O}(\text{poly}\log n)$ EPR pairs from our low entropy pseudorandom state in Section \ref{construction: tunable PRS} (where the entropy is taken to be $\text{poly}\log n$ across the chosen bipartition.) So, just by looking at the number of EPR pairs, we can distinguish between these two states, which breaks the pseudo-entanglement proof. 
\end{proof}

\subsection{Applications to property testing: An overview}

To motivate our results, consider the two following tasks.
\begin{task}
\label{applications: task 1}
Efficiently estimate the largest $t$ eigenvalues of an $n$ qubit mixed state $\rho \in \mathbb{C}^{2^n \times 2^n}$ to $\epsilon = \frac{1}{2^{\mathcal{O}(\text{poly}\log n)}}$ in additive error, starting from $\rho^{\otimes m}$. 
\end{task}

\begin{task}
\label{applications: task 2}
Efficiently estimate whether the Schmidt rank of an $n$ qubit pure state $\ket{\psi}$ is at most $2^{\mathcal{O}(\text{poly}\log n)}$, across an equipartition of qubits, starting from $\ket{\psi}^{\otimes m}$.
\end{task}
\noindent Note that for Task \ref{applications: task 1}, when $t = \omega (\text{poly}(n))$, by a Holevo bound, $m = \omega (\text{poly}(n))$. For $t = \mathcal{O}(\text{poly}(n))$, there could potentially be algorithms for which $m$ is polynomially bounded. However, from the collision bound from quantum query complexity \cite{Aaronson2004}, 
\begin{equation}
\label{two states}
m = 2^{\Omega(\text{poly} \log n)/3},
\end{equation}
for both Task \ref{applications: task 1} and Task \ref{applications: task 2}. Here is the proof sketch. Consider two states,
\begin{equation}
    |\psi_f \rangle = \frac{1}{\sqrt{2^n}} \sum_{x \in \{0, 1\}^n} \ket{x} \ket{f(x)},
\end{equation}
\begin{equation}
    |\psi_g \rangle = \frac{1}{\sqrt{2^n}} \sum_{x \in \{0, 1\}^n} \ket{x} \ket{g(x)},
\end{equation}
where $f$ is a random $1$-to-$1$ function, and $g$ is a random $2^{n - \text{poly}\log n}$-to-$1$ function. Then, if $m$ does not satisfy \eqref{two states}, this violates the quantum collision lower bound from query complexity between a $1$-to-$1$ and a $2^{n - \text{poly}\log n}$-to-$1$ function \cite{Aaronson2004}. However, even though the proof holds, note that neither $|\psi_f \rangle$ or $ |\psi_g \rangle$ has a polynomial sized circuit description. 

Our pseudorandom constructions allow us to boost the lower bound in \eqref{two states} to states that have an efficient description \footnote{Although not explicitly studied, previous pseudoentangled state constructions \cite{gheorghiu2020estimating} also imply such lower bounds, but for inverse exponentially small $\epsilon$ or exponentially large Schmidt rank. So, our lower bounds are stronger.}.

\subsubsection{Lower bound on eigenvalue estimation for efficiently preparable states}
\begin{task}
\label{applications: task 3}
Efficiently estimate the largest $t$ eigenvalues of an $n$ qubit mixed state $\rho \in \mathbb{C}^{2^n \times 2^n}$ to $\epsilon = \frac{1}{2^{\mathcal{O}(\text{poly}\log n)}}$ in additive error starting from $\rho^{\otimes m}$ with high probability, where it is promised that $\rho$ has a polynomial sized circuit description \footnote{The circuit is allowed to have trace-out gates.}. 
\end{task}

\begin{lemma}
    For Task \ref{applications: task 3} with $\epsilon = \frac{1}{2^{\mathcal{O}(\poly\log n)}}$ and $t = \mathcal{O}(\poly(n))$, assuming the existence of quantum secure one-way functions $$m = \omega \left(\poly(n)\right).$$
\end{lemma}
\begin{proof}
Follows from the security of the construction in \eqref{high entropy state construction} and \eqref{low entropy construction}. If we can perform Task \ref{applications: task 3} with $\mathcal{O}\left( \poly(n) \right)$ copies, it means we can distinguish our high-entropy pseudorandom state \eqref{high entropy state construction} from our low-entropy pseudorandom state \eqref{low entropy construction}.
\end{proof}

\begin{remark}
Note that an upper bound for Task \ref{applications: task 3} is given in \cite{tomography}. They show
\begin{equation}
    m = \mathcal{O}\left(t^2/\epsilon^2 \right).
\end{equation}
For $\epsilon = \frac{1}{2^{\mathcal{O}(\text{poly}\log n)}}$ and $t = \poly(n)$, 
\begin{equation}
    m = 2^{\mathcal{O}(\poly\log n)}.
\end{equation}
\end{remark}

\subsubsection{Lower bound on estimating the Schmidt rank for efficiently preparable states}
\begin{task}
\label{schmidt rank}
Efficiently estimate whether an $n$ qubit pure state $\ket{\psi}$ has Schmidt rank at most $2^{\mathcal{O}(\poly\log n)}$, across an equipartition of qubits, starting from $\ket{\psi}^{\otimes m}$ with high probability, where it is promised $\ket{\psi}$ has a polynomial sized circuit description.
\end{task}
\begin{lemma}
For Task \ref{schmidt rank}, assuming the existence of quantum secure one way functions, $m = \omega \left(\poly(n)\right)$.
\end{lemma}
\begin{proof}
Follows from the security of the construction in \eqref{high entropy state construction} and \eqref{low entropy construction}. If we can efficiently determine whether a state has Schmidt rank at most $r$ for $r = 2^{\mathcal{O}(\text{poly} \log n)}$, with polynomially many copies of $\ket{\psi}$, then we can use that algorithm to distinguish our high-entropy pseudorandom state \eqref{high entropy state construction}, which has Schmidt rank $2^{\Omega(n)}$, from a tunable-entropy pseudorandom state \eqref{low entropy construction} , which has Schmidt rank $2^{\mathcal{O}({r})}$ \footnote{Although our entropy calculations are based on von Neumann entropy, we implicitly also have upper and lower bounds for another entropy measure -- the logarithm of the Schmidt rank, which, by virtue of our construction, is just the logarithm of the rank of the high entropy matrix $A$ in \eqref{high entropy matrix} or the tunable entropy matrix $B$ in \eqref{lowentropymatrix}. From there, we can get corresponding bounds on the Schmidt rank of our constructed states.}.
\end{proof}

\begin{remark}
Note that an upper bound for Task \ref{schmidt rank}, of $m = \mathcal{O}(r)$, is proven in \cite{childs2007weak}. So, when $r$ is superpolynomially large, $m$ is also superpolynomially bounded.
\end{remark}


\subsection{Improved lower bound for testing matrix product states}\label{sec:mps}

Note that pseudorandom states have interesting connections to the learnability of matrix product states, as discussed in the Appendix \ref{matrix product states}. Using pseudoentanglement, we can make these connections much stronger. Specifically, optimally quasi--area law pseudoentanglement means an improved lower bound to the number of copies required to test a matrix product state. 

We show that it is difficult to test if a state is an $n$--qubit MPS with bond dimension $r$, or far from such a state, using fewer than $\Omega(\sqrt{r})$ copies of the state, in either information-theoretic or computational settings. 

\subsubsection{Previous work}

This problem was previously studied in a very recent work by \cite{soleimanifar2022testing}, who proved a lower bound of $\Omega(\sqrt{n})$. The lower bound on \cite{soleimanifar2022testing} had no dependence on $r$. The authors also prove an upper bound of $\mathcal{O}(n r^2)$: so, their bounds are significantly loose when $r$ is at least superpolynomially large.

\subsubsection{Definitions}

First, we will introduce some definitions.

\begin{definition}[Matrix product state, {\cite[Definition 1]{soleimanifar2022testing}}]\label{MPS-defn}
A quantum state $\ket{\psi}\in\C^{d_1}\otimes\cdots\otimes\C^{d_n}$ consisting of $n$ qudits is a matrix product state with bond dimension $r$ if it can be written as
\begin{align}
\ket{\psi_{1,\ldots,n}}=\sum_{i_1\in[d_1],\ldots,i_n\in[d_n]}\Tr[A_{i_1}^{(1)}\cdots A_{i_n}^{(n)}]\cdot\ket{i_1\cdots i_n},
\end{align}
where each matrix $A_j^{(i)}$ is an $r\times r$ complex matrix, for $i\in[n]$ and $j\in[d_i]$. We write $\MPS_n(r)$ for the set of such states, or more simply $\MPS(r)$ when the dependency on $n$ is clear from the context.

Further, for any state $\ket{\phi}\in\C^{d_1}\otimes\cdots\otimes\C^{d_n}$, the distance of $\ket{\phi}$ to the set $\MPS(r)$ is defined as
\begin{align}\label{eqn:Dist_r-defn}
\Dist_r(\ket{\phi})=\min_{\ket{\psi}\in\MPS(r)}\sqrt{1-|\left<\psi|\phi\right>|^2}.
\end{align}
\end{definition}

\noindent \cite{soleimanifar2022testing} also introduced the concept of $\MPS(r)$ tester.

\begin{definition}[$\MPS(r)$ tester]\label{defn:MPS-tester}
An algorithm $\mathcal{A}$ is a property tester for $\MPS(r)$ using $m=m(n,r,\delta)$ copies if, given $\delta>0$ and $m$ copies of $\ket{\psi}\in\C^{d_1}\otimes\cdots\otimes\C^{d_n}$, it acts as follows.
\begin{itemize}
\item (Completeness) If $\ket{\psi}\in\MPS(r)$, then
\begin{align}
\Pr[\mathcal{A}\text{ accepts given }\ket{\psi}^{\otimes m}]\geq\frac{2}{3}.
\end{align}
\item (Soundness) If $\Dist_r(\ket{\psi})\geq\delta$, then 
\begin{align}
\Pr[\mathcal{A}\text{ accepts given $\ket{\psi}^{\otimes m}$}]\leq\frac{1}{3}.
\end{align}
\end{itemize}
\end{definition}

\subsubsection{Our results}

Following the language of Definition~\ref{defn:MPS-tester}, \cite{soleimanifar2022testing} showed that an $\MPS(r)$ tester using $m=O(nr^2/\delta^2)$ copies of the unknown state $\ket{\psi}$ can be constructed, while any $\MPS(r)$ tester must use at least $\Omega(n^{1/2}/\delta^2)$ copies of $\ket{\psi}$. 

In this work, we show that this lower bound can be improved to order $\Omega(\sqrt{r})$, which may scale exponentially in terms of $n$. In particular, we prove the following theorem.

\begin{theorem}\label{thm:MPS-testing-lower}
\label{theoremMPS}
Following the language of Definition~\ref{defn:MPS-tester}, for any $r\leq 2^{n/8}$ and $\delta\leq\frac{1}{\sqrt{2}}$, testing whether a state $\ket{\psi}\in\C^{\otimes n}$, is in $\MPS(r)$ requires $\Omega\big(\sqrt{r}\big)$ copies of $\ket{\psi}$.
\end{theorem}

Before proving the theorem, we first states some useful results.

\begin{fact}[\cite{Vidal_2004}]\label{fact:subset-to-MPS}
For any state $\ket{\psi}\in\C^{\otimes n}$ and any partition of the state into two parts $A$ and $B$, we denote
\begin{align}
\chi_A(\ket{\psi})\coloneqq\rank(\rho_A),\qquad\rho_A(\ket{\psi})\coloneqq\Tr_B(\ket{\psi}\bra{\psi})
\end{align}
and
\begin{align}
\chi(\ket{\psi})\coloneqq\max_A \chi_A(\ket{\psi}),
\end{align}
where the maximum is taken over all possible partitions. Then,
\begin{align}
\ket{\psi}\in\MPS(\chi(\ket{\psi})).
\end{align}
\end{fact}

\begin{lemma}[Young-Eckart Theorem, \cite{eckart1936approximation}]\label{lem:Young-Eckart}
Consider a bipartite state $\ket{\psi}\in\C^{d_1}\otimes\C^{d_2}$ with $d_1\geq d_2$ and let
\begin{align}
\ket{\psi}=\sum_{i=1}^{d_2}\sqrt{\lambda_i}\ket{a_i}\ket{b_i}
\end{align}
be its Schmidt decomposition, where $\lambda_1\geq\cdots\geq\lambda_{d_2}$. Then,
\begin{align}
\Dist_r(\ket{\psi})=\sqrt{1-\sum_{i=1}^r\lambda_i}.
\end{align}
\end{lemma}

\noindent Equipped with Fact~\ref{fact:subset-to-MPS} and Lemma~\ref{lem:Young-Eckart}, we are now ready to prove Theorem~\ref{thm:MPS-testing-lower}.

\begin{proof}[Proof of Theorem~\ref{thm:MPS-testing-lower}]

Note that any $\MPS(r)$ tester defined in Definition~\ref{defn:MPS-tester} can distinguish with success probability at least $2/3$ between any two ensembles of quantum states, one only containing matrix product states with bond dimension at most $r$, the other only containing states whose distance to $\MPS(r)$ is at least $\frac{1}{\sqrt{2}}$. In this proof, we explicitly construct such two ensembles and demonstrate that any quantum algorithm having less than $O(\sqrt{r})$ copies of a state $\ket{\psi}$ cannot determine with success probability $2/3$ which of the two ensembles $\ket{\psi}$ is in, thus establishing an $\Omega(\sqrt{r})$ lowerbound for MPS testing.

We use $\mathcal{E}_r$ to denote the ensemble of subset phase states with subset size $r$ and random phase. Quantitatively,
\begin{align}
\mathcal{E}_r\coloneqq\big\{\ket{\psi_{f,S}}\,\big|\,|S|=r\big\},
\end{align}
where $\ket{\psi_{f,S}}$ is defined in \eqn{subset-phase}. Observe that for any partition of any subset phase state $\ket{\psi_{f,S}}$ into two parts $A$ and $B$, the rank of the corresponding reduced density matrix $\rho_A\coloneqq\Tr_B(\ket{\psi}\bra{\psi})$ is upper bounded by $r$. Then by Fact~\ref{fact:subset-to-MPS}, we have
\begin{align}
\mathcal{E}_r\subseteq\MPS(r).
\end{align}

Next, we construct an ensemble of quantum states that are far from $\MPS(r)$. Specifically, we consider the ensemble $\mathcal{E}_{\mathrm{phase}}$ consisting of phase states with random phases,
\begin{align}
\mathcal{E}_{\phase}\coloneqq\left\{\ket{\psi_f}\right\},\qquad\ket{\psi_f}=\frac{1}{2^n}\sum_x(-1)^{f(x)}\ket{x}.
\end{align}
For any cut $(\mathsf{A},\mathsf{B})$ of $n$ qubits such that $|\mathsf{A}|=|\mathsf{B}|=n/2$, by \eqn{phase-state-lower-bound} and Markov's inequality, we know that
\begin{align}
\Pr_{\ket{\psi_f}\leftarrow\mathcal{E}_{\phase}}\left[\|\rho_{\mathsf{A}:\mathsf{B}}\|_F\leq\frac{1}{2^{n/4}}\right]\geq 1-2^{-n/4}.
\end{align}
That is to say, if we uniformly randomly select a state $\ket{\psi_f}$ from $\mathcal{E}_{\phase}$, with probability at least $1-2^{-n/4}$, its Schmidt decomposition
\begin{align}
\ket{\psi_f}=\sum_{i=1}^{d_2}\sqrt{\lambda_i}\ket{a_i}\ket{b_i}
\end{align}
satisfies
\begin{align}
\sum_{i=1}^{d_2}\lambda_i^2\leq 2^{-n/4},
\end{align}
where $\lambda_1\geq\cdots\geq\lambda_{d_2}$. By Cauchy's ineuqality,
\begin{align}
\sum_{i=1}^r\lambda_i\leq\sqrt{r\cdot 2^{-n/4}}\leq 2^{-n/8}\leq\frac{1}{2}.
\end{align}
Then by Lemma~\ref{lem:Young-Eckart}, we have
\begin{align}\label{eqn:large-distance-high-probability}
\Pr_{\ket{\psi_f}\leftarrow\mathcal{E}_{\phase}}\left[\Dist_r(\ket{\psi_f})\geq\frac{1}{\sqrt{2}}\right]\geq 1-2^{-n/4}.
\end{align}
We define $\mathcal{E}_{\phase}'$ to be the set of phase states that are at least $1/\sqrt{2}$-far from $\MPS(r)$. In particular,
\begin{align}
\mathcal{E}_{\phase}'\coloneqq\left\{\ket{\psi_f}\Big|\Dist_r(\ket{\psi_f})\geq\frac{1}{\sqrt{2}}\right\}.
\end{align}
Based on \eqn{large-distance-high-probability}, we have
\begin{align}
\TD\left(\underset{\ket{\psi}\leftarrow\mathcal{E}_{\phase}'}{\E}\big[\ket{\psi}\bra{\psi}^{\otimes t}\big],\underset{\ket{\phi}\leftarrow\mathcal{E}_{\phase}}{\E}\big[\ket{\phi}\bra{\phi}^{\otimes t}\big]\right)\leq 2^{-n/4}
\end{align}
for any $t$. Further, since
\begin{align}
\mathsf{TD}\left(\underset{\ket{\psi}\leftarrow\mathcal{E}_{\phase}}{\E}\left[\ket{\psi}\bra{\psi}^{\otimes t}\right], \underset{\ket{\phi}\leftarrow\mathscr{H}(\C^N)}{\E}\left[\ket{\phi}\bra{\phi}^{\otimes t}\right]  \right)  < O\Bigg(\frac{t^2}{2^n}\Bigg),
\end{align}
we can derive that
\begin{align}
\mathsf{TD}\left(\underset{\ket{\psi}\leftarrow\mathcal{E}_{\phase}'}{\E}\left[\ket{\psi}\bra{\psi}^{\otimes t}\right], \underset{\ket{\phi}\leftarrow\mathscr{H}(\C^N)}{\E}\left[\ket{\phi}\bra{\phi}^{\otimes t}\right]  \right)  < O\left(\frac{t^2}{2^{n/4}}\right),
\end{align}
where $\mathscr{H}(\C^N)$ denotes the ensemble of Haar random states in the Hilbert space with dimension $N=2^n$. Moreover, by Theorem~\ref{thmref:SPS-security},
\begin{align}
    \mathsf{TD}\left(\underset{\ket{\psi}\leftarrow\mathcal{E}_r}{\E}\left[\ket{\psi}\bra{\psi}^{\otimes t}\right], \underset{\ket{\phi}\leftarrow\mathscr{H}(\C^N)}{\E}\left[\ket{\phi}\bra{\phi}^{\otimes t}\right]  \right)  < O\left(\frac{t^2}{r}\right),
\end{align}
which further leads to
\begin{align}
\mathsf{TD}\left(\underset{\ket{\psi}\leftarrow\mathcal{E}_r}{\E}\left[\ket{\psi}\bra{\psi}^{\otimes t}\right], \underset{\ket{\phi}\leftarrow\mathcal{E}_{\phase}'}{\E}\left[\ket{\phi}\bra{\phi}^{\otimes t}\right]  \right)  < O\left(\frac{t^2}{2^{n/4}}\right)+O\left(\frac{t^2}{r}\right)<O\left(\frac{t^2}{r}\right).
\end{align}
Hence, any quantum algorithm distinguishing between $\mathcal{E}_r$ and $\mathcal{E}_{\phase'}$ with $\Omega(1)$ success probability requires at least $t=\Omega(\sqrt{r})$ copies of the unknown state, which is also the lower bound for MPS testing.
\end{proof}

\begin{remark}
    By using appropriate security conjectures and appropriate cryptographic primitives to efficiently instantiate our pseudorandom functions --- for example, by one--way functions which are secure upto subexponential time against quantum adversaries \footnote{One way functions based on LWE are conjectured to have this property.} --- we can get the same lower bounds as in Theorem \ref{theoremMPS} in the computational setting, when the matrix product state under consideration is guaranteed to have an efficient description. 
\end{remark}

\subsection{Applications to quantum gravity}
\label{subsec:adscft}

Another application of our result is to quantum gravity.
The AdS/CFT correspondence \cite{maldacena1999large} is one of the leading candidates for a theory of quantum gravity.
It postulates a duality between a theory of quantum gravity in anti-de Sitter space (AdS) and simple quantum mechanical theory (namely, a conformal field theory (CFT)).
The AdS/CFT ``dictionary'' maps states in one theory to the other, and through this dictionary observables and states are mapped from one theory to the other.
In this way one can study properties of quantum gravity via studying a simpler quantum mechanical system.
This has led to a series of remarkable results connecting quantum gravity with topics in quantum information such as quantum error correction \cite{almheiri2015bulk}, quantum tensor networks~\cite{pastawski2015holographic}, the Eastin-Knill Theorem~\cite{faist2020continuous}, and quantum circuit complexity~\cite{susskind2016computational}.
There has even the suggestion that future quantum computers might shed light into quantum gravity \cite{brown2019quantum,nezami2021quantum}.

Recently, Bouland, Fefferman and Vazirani \cite{bouland2019computational} showed that the AdS/CFT dictionary might be exponentially complex to compute, even for a quantum computer.
This stands in sharp contrast to other dualities in computer science, such as LP and SDP duality, which are efficiently computable. 
Their argument used the fact that certain information about the geometry of the gravitational theory -- in particular information about the interior of a wormhole -- seems to be pseudorandomly scrambled in the quantum theory, in a manner analogous to a block cipher such as DES.
Therefore efficient computation of the dictionary to reconstruct wormhole interiors allows one to break certain forms of cryptography, which is not believed to be tractable for quantum computers.
Their result seems to challenge the quantum Extended Church-Turing thesis, i.e. the conjecture that all physical processes are efficiently simulable by a universal quantum computer.

The arguments of \cite{bouland2019computational} require the presence of a black hole. It is natural to ask if similar arguments might show the dictionary might be exponentially complex in more general gravitational geometries.
Indeed, Susskind \cite{susskind2020horizons} has suggested this might not be possible, i.e. that without black holes (and outside of the event horizon of black holes) the dictionary is easy to compute.

A potential starting point to investigate this issue is the more general connection between entanglement and geometry in AdS/CFT.
It is believed that in AdS/CFT, the entanglement entropy between certain regions of the quantum theory is directly proportional to certain geometrical quantities in the gravity theory (namely, the shortest geodesic between corresponding boundary points in the spacetime), up to small corrections, via the Ryu-Takanayagi formula \cite{ryu2006holographic}.
Therefore the entanglement entropy of CFT is directly connected to the geometry.

It is natural to ask if this connection between entanglement and geometry already implies the dictionary is exponentially hard to compute. 
If so this would provide complementary evidence to \cite{bouland2019computational} for the exponential complexity of the dictionary, and potentially remove the need for a black hole in those arguments.
This was precisely the suggestion of Hoban and Gheorghiu in their construction of what we call pseudoentanglement \cite{gheorghiu2020estimating}.
By showing that entanglement entropy is exponentially difficult to compute, even on average, they argued this was evidence for the exponential complexity of the dictionary.
Their work left open many future directions to further develop this argument.
Most salient is whether or not it is possible to create pseudoentanglement within the subset of holographic states, i.e. states for which the AdS/CFT dictionary is well-defined. 
Such states exhibit many atypical features, for example sub-volume law (but super-area law) entanglement, which are not properties of Hoban and Gheorghiu's construction.

Our result strengthens the case for this argument, as we show that it is possible to construct pseudorandomness or pseudoentanglement with subvolume law entanglement.
This is a necessary but not sufficient condition to construct pseudorandomness and pseudoentanglement within the domain of validity of AdS/CFT, and therefore paves the way to potentially constructing pseudoentanglement with holographic entanglement structures.
More speculatively, we believe the ``tunable'' nature of our construction might be useful for constructing pseudoentanglement with varying geometries, which might form the basis of a future challenge to the quantum ECT without the presence of black holes.
We leave further development of this argument to future work.

Finally, we note this line of argument is complementary to very recent work of Aaronson and Pollack \cite{aaronsonpollack22}, who showed that given as input a list of entropies of a CFT state obeying certain conditions, that there is an efficient algorithm to produce a bulk state with the corresponding entropies. In contrast, our work shows that it is difficult to produce the list of entanglement entropies given as input a quantum state. If our argument could be made holographic, this might show the difficulty of the dictionary stems from the ability of quantum states to hide their entanglement entropies. 

\section{Future directions}
We close with some natural future directions that are left open by this work.
\begin{enumerate}
    \item{A natural question left open in this work is to understand the importance of the random phases in our subset phase state construction.  That is, consider states of the form:
\begin{equation}
    |\psi_S\rangle = \frac{1}{|S|} \sum_{x \in S} \ket{x}
\end{equation}
Are these states pseudorandom and pseudoentangled if $S \in \{0, 1\}^n$ is a pseudo-randomly chosen subset of appropriate size? Note that this is similar to the construction in Section \ref{main construction} which we discussed in this paper without the pseudorandom phases. }
\item{Do other families of quantum states achieve tightly tuned pseudoentanglement across any cut?  We discuss two variants in Appendix \ref{1D area law} and \ref{2D area law}, where we prove a pseudo--area law scaling of entanglement. However, we only prove an upper bound on the entanglement entropy: It remains to see if our upper bound is tight.}
\item{Finally, are there further applications of pseudoentanglement to cryptography, complexity theory, and quantum computing?}
\end{enumerate}

\section{Acknowledgments}
We thank Jordan Docter, Tudor Giurgica-Tiron, Nick Hunter-Jones, and Wilson Nguyen for helpful discussions. 
B.F. and S.G. acknowledge support from AFOSR (FA9550-21-1-0008).
This material is based upon work partially
supported by the National Science Foundation under Grant CCF-2044923 (CAREER) and by the U.S. Department of Energy, Office of Science, National Quantum Information Science Research Centers (Q-NEXT).
This research was also supported in part by the National Science Foundation under Grant No. NSF PHY-1748958.  
A.B. and B.F. were supported in part by the DOE QuantISED grant DE-SC0020360.  
A.B. and C.Z. were supported in part by the AFOSR under grant FA9550-21-1-0392.
A.B. was supported in part by the U.S. DOE Office of Science under Award Number DE-SC0020266. U.V acknowledges the Vannevar Bush faculty fellowship N00014-17-1-3025, and was supported by DOE NQISRC QSA grant FP00010905,
QSA grant FP00010905, and NSF QLCI Grant No. 2016245.
Z.Z. was supported in part by a Stanford School of Engineering Fellowship.

\bibliographystyle{alpha}
\newcommand{\etalchar}[1]{$^{#1}$}

\newpage 

\appendix

\section{Low entanglement across a fixed cut within the JLS phase state construction}
\label{app:singlecutjls}

In a previous version of this work, we presented a method to have a PRS with low entanglement across a fixed cut.  This can be done using the binary pseudorandom phase state construction, proposed by \cite{Ji2018} and proven to be pseudorandom by \cite{brakerski2019pseudo}. We will later generalize this to low entanglement with 1D pseudo-area law entanglement, and then finally describe a construction that has low entanglement across any cut, using a new idea based on subset states.

\subsection{Sketch of result}
Consider states of the form:
\[ \displaystyle\sum_{x\in\{0,1\}^n} (-1)^{f(x)} \ket{x}\]
where $f(x):\{0,1\}^n\rightarrow \{0,1\}$ is a pseudorandom function.
One might expect that typical PRFs would yield maximally entangled phase states.
Let us first assume this is the case (we will describe how to construct such a PRF shortly).
We will now describe how to reduce its entanglement entropy to any desired amount -- thus giving a PRS construction with low entanglement.

To do this, it turns out to be helpful for the analysis to view this PRF as defining a ``pseudorandom matrix''\footnote{That is, simply write out the $2^n$-length truth table of the function in the form of a $2^{n/2}$ by $2^{n/2}$ matrix.} $A_{i,j} \in \{\pm1\}^{2^{n/2}\times 2^{n/2}}$ where the indices run over $\{0,1\}^{n/2}$, so the resulting state is
\[ \displaystyle\sum_{i,j\in\{0,1\}^{n/2}} A_{i,j} \ket{i}\ket{j}=\sum_{i,j\in\{0,1\}^{n/2}}(-1)^{f(i,j)}\ket{i}\ket{j}\]
Now the reduced density matrix on the first half of the qubits is $\rho = \frac{1}{2^{n}} A A^T$, so we wish to minimize
$\mathsf{S}(\rho)= \mathsf{S}\left(\frac{1}{2^{n}} A A^T\right)$.

We consider augmenting the pseudorandom matrix $A$ with certain \emph{row operations} that have measurable effects on the entanglement entropy of the corresponding reduced density matrix $\rho$. In particular, imagine taking $A$, picking a small subset of the rows, and \emph{repeating} those rows many times (overwriting some of the rows of $A$ in the process) to create a matrix $A'$. 
One can easily see this would reduce the entanglement entropy of the resulting state.
The reason is that this reduces the rank of the matrix -- and hence the rank of $\rho\propto A A^T$.
While this reduces the entanglement, it is unclear if this preserves the pseudorandomness of the matrix.

To reduce the rank of $A$ while preserving pseudorandomness, we use a pseudorandom $2^{n/2-k}\rightarrow 1$ function $g:\{0,1\}^{n/2}\rightarrow \{0,1\}^{n/2}$ to ``select'' the subset of repeated rows.  More formally, we define $g=f_1(f_2(x) \mod 2^{k})$ where $f_1$ and $f_2$ are pseudorandom permutations.  It's not difficult to see that $g$ is $2^{n/2-k}$-to-$1$, and we prove that it is pseudorandom using a series of hybrid arguments, see Section \ref{subsection: low entropy}. 
Then we construct a new matrix $A'$ which uses $g$ to \emph{subsample} the rows of $A$ as follows: \[A'_{i,j} = A_{g(i),j}\]

As the pseudorandom function $g$ is $2^{n/2-k}\rightarrow 1$, the rank of $A'$ is at most $2^{k}$ -- so the entanglement entropy of the corresponding reduced density matrix $\rho$ is at most $k$. The problem then reduces to determining how small one can set $k$ in our $2^{n/2-k} \rightarrow 1$ function without destroying its post-quantum security.
Clearly the range of the function $g$ must be at least superpolynomial -- otherwise collisions would be abundant, and by just picking random rows of $A'$ and checking if they are equal\footnote{Say, by preparing phase states corresponding to rows of $A$ and using a SWAP test.}, one would find collisions with inverse polynomial probability. 
Surprisingly the value of $k$ can be set all the way to this limit (e.g. $k=\omega(\log n)$, which produces entanglement entropy $\omega(\log n)$).
This follows from the strength of the collision bound for distinguishing random functions with small range from PRPs \cite{zhandry2012construct, zhandry2013note}.
The proof turns out to be quite involved, putting together a number of different techniques due to Zhandry \cite{zhandry2012construct}.  The most subtle issue at play is that black-box security proofs used in the collision bound are transferable to the white-box settings of pseudorandom state constructions -- essentially because in pseudorandom states the access to the PRF is limited to access to copies of the phase state, which are preparable in a black-box manner.  We defer a full technical proof to Appendix \ref{app:singlecutjls}.

\section{Technical overview of the low entanglement JLS pseudorandom state construction}

For a function $f(i, j)$, from $\{0, 1\}^n$ to $\{0, 1\}$, consider the phase state given by
\begin{equation}
\label{construction1}
    \ket{\psi_f} = \frac{1}{\sqrt{2^{n}}} \sum_{i, j \in \{0, 1\}^{n/2}} (-1)^{f(i, j)} \ket{i, j}.
\end{equation}
It was proven in \cite{brakerski2019pseudo} that if $f(i, j)$ is a pseudorandom function, then the construction in \ref{construction1} is a pseudorandom quantum state. When discussing pseudorandom quantum states, we will often refer to \emph{pseudorandom matrices} which we define below.
\subsection{Pseudorandom matrices}
For a pseudorandom function $f(i, j)$, define a matrix $A^{(f)}$ such that
\begin{equation}
    A^{(f)}_{(i, j)} = (-1)^{f(i, j)},
\end{equation}
for all $i, j \in \{0, 1\}^{n/2}$. We call $A^{(f)}$ a \emph{pseudorandom matrix}. Often, for ease of notation, when the function is clear from the context, we will drop the superscript $f$ and just write $A$ to denote the corresponding matrix. Note that, by definition, a pseudorandom matrix is computationally indistinguishable from a truly random matrix \footnote{The entries of a truly random matrix come from a truly random function $r$.} when given black box query access to the entries of the matrix.  

The motivation behind defining pseudorandom matrices is that they help us neatly characterize the reduced density matrix across each partition in \eqref{construction1}. To illustrate, the reduced density matrix across the first $n/2$ qubits is

\begin{align}
\rho[1] &= \frac1{2^n} \left(\sum_{i \in \{ 0, 1\}^{n/2} } \sum_{j \in \{ 0, 1\}^{n/2} }   \sum_{k \in \{ 0, 1\}^{n/2} } \sum_{l \in \{ 0, 1\}^{n/2} } (-1)^{f(i, j)+f(k, l)} \Tr_2( \ket{i} \ket{j} \bra{k}\bra{l} )\right) \\
& = \frac1{2^n} \left(\sum_{i \in \{ 0, 1\}^{n/2} } \sum_{j \in \{ 0, 1\}^{n/2} }   \sum_{k \in \{ 0, 1\}^{n/2} } (-1)^{f(i, j)+f(k, j)}  \ket{i} \bra{k} \right) \\
& = \frac1{2^n} \left(\sum_{i \in \{ 0, 1\}^{n/2} } \sum_{j \in \{ 0, 1\}^{n/2} }   \sum_{k \in \{ 0, 1\}^{n/2} } A^{}_{i, j}A^{}_{k, j}\ket{i} \bra{k} \right) \\
& = \frac1{2^n} \left(\sum_{i \in \{ 0, 1\}^{n/2} } \sum_{j \in \{ 0, 1\}^{n/2} }   \sum_{k \in \{ 0, 1\}^{n/2} } A^{}_{i, j}A^{}_{k, j}\ket{i}\langle j| j \rangle \bra{k} \right) \\
& = \frac1{2^n} \left(\sum_{i \in \{ 0, 1\}^{n/2} }  \sum_{j \in \{ 0, 1\}^{n/2} } A^{}_{i, j}\ket{i}\langle j| \right) \left(\sum_{j \in \{ 0, 1\}^{n/2} }  \sum_{k \in \{ 0, 1\}^{n/2} } A^{}_{k, j}\ket{j}\langle k| \right) \\
\label{section2:1 last line}
&=\frac{1}{2^n} AA^{\mathsf{T}}.
\end{align}

\noindent We can do a similar calculation for $\rho[2]$, the reduced density matrix across the last $n/2$ qubits, to reach \eqref{section2:1 last line}. Hereon, we will let 
\begin{equation}
\label{definition rho}
\rho = \rho[1] = \rho[2].
\end{equation}

\subsubsection{Entanglement entropy and rank}
For a $2^{n/2} \times 2^{n/2}$ density matrix $\rho$, define the entanglement entropy to be
\begin{equation}
\begin{aligned}
    \mathsf{S}(\rho) &= -\Tr(\rho \log \rho).
\end{aligned}
\end{equation}
Let $r$ be the rank of $\rho$. It holds that
\begin{equation}
\label{entropy bounds}
     -\log ||\rho||_F  \le  -\log ||\rho|| \leq \mathsf{S}(\rho) \leq \log r,
\end{equation}
by the spectral decomposition of $\rho$ and an application of Jensen's inequality. The Frobenius norm of $\rho$, given by $||\rho||_F$, is defined as
\begin{equation}
    ||\rho||_F = \left(\sum_{i=1}^r \lambda_i^2\right)^{1/2} = \left(\sum_{i=1}^{2^{n/2}} \sum_{j=1}^{2^{n/2}} |\rho_{i, j}|^2 \right)^{1/2},
\end{equation}
where $\lambda_i$ is the $i^{\text{th}}$ non-zero eigenvalue of $\rho$. Additionally note that for any matrix $A$,
\begin{equation}
    \mathrm{rank}(A) = \mathrm{rank}(A A^{\mathsf{T}}).
\end{equation}
Now, following \eqref{definition rho}, let
\begin{equation}
    \rho = \frac{1}{2^n} AA^{\mathsf{T}}.
\end{equation}
Hence,
\begin{equation}
\label{rank and entropy}
   -\log \left(\bigg|\bigg|\frac{1}{2^n} AA^{\mathsf{T}} \bigg|\bigg|_F\right) \leq S(\rho) \leq \log \mathrm{rank}(A).
\end{equation}

\subsubsection{Constructing pseudoentangled states}
\label{construction: pseudoentangled states}
In the next two sections, we will construct two different types of pseudorandom state ensembles, the first where a typical state of the family has high entanglement entropy and the second where a typical state of the family has low entanglement entropy. As, by transitive law, these states are computationally indistinguishable from each other, we call them \emph{pseudoentangled states}.

During the construction of each ensemble, we will work towards satisfying two requirements simultaneously. 

Firstly, we want the matrix $A$ that we construct to be quantum-secure pseudorandom matrix. That is, we want $A$ to be indistinguishable from a random matrix to quantum adversaries, when given black-box access to either $A$ or that random matrix. This is because, one special object an adversary can construct when given black box access to a $2^{n/2} \times 2^{n/2}$ matrix $M$ is the phase state
\begin{equation}
    \ket{\psi_{M}} = \frac{1}{\sqrt{2^n}} \sum_{i, j \in \{0, 1\}^{n/2}} (-1)^{M(i, j)} \ket{i, j}.
\end{equation}
In fact, when given black box access to $M$, an adversary can construct polynomially many copies of $\ket{\psi_{M}}$. So, if $A$ is black-box indistinguishable from a random matrix, it implies that polynomially many copies of the phase state $\ket{\psi_{A}}$ must be computationally indistinguishable from polynomially many copies of the phase state $\ket{\psi_{R}}$, where $R$ is a random matrix. If they were distinguishable, it would break the black box indistinguishability result for $A$ -- given black box access, the adversary could just prepare a phase state and run the distinguisher for phase states! Now, from \cite{Ji2018}, polynomially many copies of $\ket{\psi_{R}}$ are statistically close to polynomially many copies of a Haar random state. Hence, by transitive law, $\ket{\psi_{A}}$ is a pseudorandom quantum state. 

Secondly, following \eqref{rank and entropy}, we want the Frobenius norm of $A$ to be large when we are constructing a state with high entanglement entropy. To that end, we will  lower bound
\begin{equation}
    -\log \left(\bigg|\bigg|\frac{1}{2^n} AA^{\mathsf{T}} \bigg|\bigg|_F\right)
\end{equation}
in the high entropy construction. In the low entropy construction, the very fact that our construction gives a pseudorandom quantum state means that it satisfies the generic entropy lower bound for pseudorandom states in \cite{Ji2018} -- this will be sufficient to match our upper bound. So, we do not need to lower bound the Frobenius norm. In lieu of that, we will instead upper bound the rank of $A$ to show that the entanglement entropy is small and matches the generic lower bound.

\subsection{Pseudorandom matrix with high entanglement entropy}
\label{subsection: high entropy}
We will now construct a family of pseudorandom states with high entanglement entropy. To that end, we will construct an ensemble of $2^{n/2} \times 2^{n/2}$ pseudorandom matrices with each entry $\in \{ -1, 1\}$. All except for a negligible fraction of matrices in this ensemble have high entanglement entropy. That is, our construction will guarantee that for a matrix $A$ picked from our ensemble,
\begin{equation}
\mathsf{S}\left(\frac{1}{2^{n}} AA^{\mathsf{T}}\right) = \Omega(n),
\end{equation}
with very high probability over the choice of $A$.
\subsubsection{High level overview of the construction}
\label{entropy: random matrix}
Before delving into the construction of a pseudorandom matrix, let us study some properties of a $2^{n/2} \times 2^{n/2}$ random matrix $R$, with entries in $\{1, -1\}$ to hone our intuition. Specifically, let us try to lower bound
\begin{equation}
\label{frobenius norm 2}
    -\log \left(\bigg|\bigg|\frac{1}{2^n} RR^{\mathsf{T}} \bigg|\bigg|_F\right).
\end{equation}
For \eqref{frobenius norm 2} to be $\Omega(n)$, the Frobenius norm of $R$ should be inverse exponentially small in $n$. Let us try to prove this, in expectation over the choice of $R$.

	\begin{align}\label{eqn:phase-state-lower-bound}
		&\mathrm{E} \left [ \left \|\frac{1}{2^{n}} RR^{\mathsf{T}} \right \|^2_F \right] \nonumber \\
        &= \frac{1}{2^{2n}} \mathrm{E} \left [ \left \| RR^{\mathsf{T}} \right \|_F \right] \nonumber\\
		& = \frac{1}{2^{2n}} \sum_{i=1}^{2^{n/2}}  \sum_{j=1}^{2^{n/2}} \mathrm{E} \left[\left ( \sum_{k=1}^{2^{n/2}}R_{ik}\cdot R_{jk}\right)^2 \right] \nonumber \\
		& = \frac{1}{2^{2n}} \sum_{i=1} ^{2^{n/2}}  \mathrm{E} \left[\left ( \sum_{k=1}^{2^{n/2}}R_{ik}\cdot R_{ik} \right)^2 \right] + \frac{1}{2^{2n}} \sum_{i \ne j, i, j = 1}^{2^{n/2}} \mathrm{E} \left[\left ( \sum_{k=1}^{2^{n/2}}R_{ik}\cdot R_{jk}\right)^2 \right] \nonumber \\
		& = \frac{1}{2^{n/2}} + \frac{1}{2^{2n}} \sum_{i \ne j, i, j = 1}^{2^{n/2}} \sum_{k=1}^{2^{n/2}}\mathrm{E} \left[\left ( R_{ik}\cdot R_{jk}\right)^2 \right] \nonumber \\
		& \le \frac{1}{2^{n/2-1}}.
	\end{align}
Now, we can use Markov's inequality to argue that the Frobenius norm of a typical random matrix $R$ is also inverse exponentially suppressed with high probability. 

Note that the only place where we utilized the fact that $R$ is a random matrix is when going from line $4$ to line $5$. Even there, we only utilized the fact that $R$ is $4$-wise independent \footnote{A random matrix is trivially $k$-wise independent, for any integer $k$.} to argue that for distinct $i, j, k, p, q, s \in \{0,1\}^{n/2}$,
\begin{equation}
\begin{aligned}
    &\mathrm{E}[R_{ik} R_{jk} R_{pq} R_{sq}] \\
    &= \mathrm{E}[R_{ik}] \cdot \mathrm{E}[R_{jk}] \cdot \mathrm{E}[R_{pq}] \cdot \mathrm{E}[R_{sq}] \\ 
    &= 0.
\end{aligned}
\end{equation}
Taking inspiration from this fact, there are two desirable properties we will seek in our construction.
\begin{itemize}
  \item It should be oracle indistinguishable from a random matrix to quantum adversaries, to be compatible with the strategy we outlined in \ref{construction: pseudoentangled states}.
    \item It should be $4$-wise independent. As we just saw, this helps us to easily argue that the Frobenius norm is inverse exponential in $n$ with high probability.
\end{itemize}
\subsubsection{The construction}
\label{A matrix}
First we consider the matrix $A$ with entries:
\begin{equation}
A_{ij}:= f(p(i, j)),
\end{equation}
where $f$ is uniformly drawn from a $4$-wise independent function family 
\begin{equation}
F= \{f :[2^{n/2}] \times [2^{n/2}] \to \{1, -1 \} \},
\end{equation}
and $p$ is uniformly drawn from $P$ -- a quantum-secure pseudorandom permutation (PRP) family. 
\subsubsection{Analysis of oracle indistinguishability}
\label{oracle indistinguishability} 
\begin{definition}
	We say a function $f :[2^{n/2}] \times [2^{n/2}] \to \{1, -1 \} $ is $\epsilon$-biased, if $$\left(\frac12 - \epsilon\right) \cdot 2^{n} \le |f^{-1}(1)| \le \left(\frac12 + \epsilon\right) \cdot 2^{n}.$$
\end{definition}
\begin{lemma} \label{lem:4wisechebyshev}
	If $F$ is a family of $4$-wise independent functions, $f$ is a random sample uniformly drawn from $F$. Then $f$ is $2^{-n/4}$-biased with high probability.
\end{lemma}
\begin{proof}
	The fact that $F$ is a family of $4$-wise independent functions implies the variance of sum of all entries of $f$ equals the sum of variances\footnote{This is true even for pairwise independent functions.}. It also implies that the function $f$ is balanced, in expectation \footnote{This is true even for $1$-wise independent functions.}. More specifically,
	\begin{align*}
	\mathrm{Var}\left[ \sum_{i=1}^{2^{n/2}}\sum_{j=1}^{2^{n/2}} f(i, j) \right] & = \sum_{i=1}^{2^{n/2}}\sum_{j=1}^{2^{n/2}} \mathrm{Var}[f(i, j)] = 2^{n}. \\
	 \mathrm{E} \left[ \sum_{i=1}^{2^{n/2}}\sum_{j=1}^{2^{n/2}} f(i, j) \right] &= 0
		\end{align*}
	By Chebyshev's inequality, we have 
	\begin{align*}
		\Pr\left[\left | \sum_{i=1}^{2^{n/2}}\sum_{j=1}^{2^{n/2}} f(i, j) \right | > 2^{\frac34 n} \right] & \le \mathrm{Var}\left( \sum_{i=1}^{2^{n/2}}\sum_{j=1}^{2^{n/2}} f(i, j) \right) \cdot 2^{-\frac32 n} \\
		& = 2^{-n/2}.
	\end{align*}
\end{proof}

\begin{definition}
	We say a random matrix is entry-wise efficiently sampleable if each entry is drawn i.i.d. from an efficiently sampleable distribution $D$.
\end{definition}

\begin{definition}
	We say a random matrix is  $\epsilon$-biased, if each entry is drawn i.i.d. from a distribution $D$ with support $\{ 1, -1 \}$, and 
	\begin{equation}
	\frac{1}{2} - \epsilon \le \underset{x \sim D}{\Pr}[x = 1] \le \frac{1}{2} + \epsilon.
	\end{equation}
\end{definition}

\begin{lemma}\label{lem:biasedmatrix}
	An entry-wise efficiently sampleable $2^{-n/4}$-biased random matrix is oracle-indistinguishable from $0$-biased random matrix.
\end{lemma}
\begin{proof}
	Let $D_1$ be a distribution with support $\{ 1, -1 \}$, and 
	\begin{equation}
	\frac{1}{2} - 2^{-n/4} \le \Pr_{x \sim D}[x = 1] \le \frac{1}{2} + 2^{-n/4},
	\end{equation}
	and $D_2$ be a uniform distribution over $\{ 1, -1 \}$. It is clear that $D_1$ and $D_2$ are statistically indistinguishable (thus computationally indistinguishable). By Theorem 4.5 of \cite{Zha21}, we immediately know that $2^{-n/4}$-biased random matrix is oracle-indistinguishable from a random matrix.
\end{proof}

\begin{corollary}\label{coro:epsindistinguishable}
	If $f :[2^{n/2}] \times [2^{n/2}] \to \{1, -1 \} $ is a fixed $2^{-n/4}$-biased function, then the ensemble $f(r(\cdot, \cdot))$ is oracle-indistinguishable from $0$-biased random matrix, where $r: [2^{n/2}] \times [2^{n/2}] \to [2^{n/2}] \times [2^{n/2}]$ is a ($0$-biased) truly random function.
\end{corollary}
\begin{proof}
	Note that entries of $r$ are i.i.d. uniformly drawn from $[2^{n/2}] \times [2^{n/2}]$. It is easy to see that ensemble $f(r(\cdot, \cdot))$, where $f$ is a fixed $2^{-n/4}$-biased function, and $r$ is a random function,  is exactly an entry-wise efficiently sampleable $2^{-n/4}$-biased random matrix, by Lemma~\ref{lem:biasedmatrix}, it is oracle-indistinguishable from a random matrix.

\end{proof}

\begin{lemma}\label{lem:permutationrandomindistinguishable}
	$f(p(\cdot, \cdot))$ and $f(r(\cdot, \cdot))$ are oracle-indistinguishable for any fixed function $f: [2^{n/2}] \times [2^{n/2}] \to \{1, -1\}$, where $p$ is drawn from a pseudorandom permutation (PRP) family and $r: [2^{n/2}] \times [2^{n/2}] \to [2^{n/2}] \times [2^{n/2}]$ is a truly random function.
\end{lemma}
\begin{proof}
	By a straightforward hybrid argument. By the definition of PRP, $f(p(\cdot, \cdot))$ and $f(p_r(\cdot, \cdot))$ are oracle-indistinguishable, where $p_r: [2^{n/2}] \times [2^{n/2}] \to [2^{n/2}] \times [2^{n/2}] $ is a truly random permutation. Next, by \cite{Zha15}, $f(p_r(\cdot, \cdot))$ and $f(r(\cdot, \cdot))$ are oracle-indistinguishable.
\end{proof} 

\begin{lemma} \label{lem:PRPindistinguishable}
	 Let $f :[2^{n/2}] \times [2^{n/2}] \to \{1, -1 \} $ be $2^{-n/4}$-biased, $f(p(\cdot, \cdot))$ is oracle-indistinguishable from $0$-biased random matrix.
\end{lemma}
\begin{proof}
	By Corollary~\ref{coro:epsindistinguishable} and Lemma~\ref{lem:permutationrandomindistinguishable}, the proof follows.
\end{proof}

\begin{lemma} \label{lem:4wiseindependent}
Let $F= \{f :[2^{n/2}] \to [2^{n/2}]\}$ be a $4$-wise independent function family, and $P = \{f :[2^{n/2}] \to [2^{n/2}]\}$ be a pseudorandom permutation family. Then $f(p(\cdot, \cdot))$ is oracle-indistinguishable from a random matrix \footnote{A random matrix has $0$ bias.}, where $f$ is drawn from $F$ uniformly at random, and $p$ is drawn from $P$ uniformly at random.
\end{lemma}
\begin{proof}
    Let $F' \subseteq F$ be the subset of $F$ such that all functions $f' \in F'$ are $2^{-n/4}$-biased. Suppose quantum algorithm $\mathcal{A}$ distinguishes  $f(p(\cdot, \cdot))$ from a ($0$-biased) random matrix with non-negligible probability. Since $f$ is $2^{-n/4}$-biased with high probability, by Lemma~\ref{lem:4wisechebyshev}, $\mathcal{A}$ must distinguishes  $f'(p(\cdot, \cdot))$ from a random matrix with non-negligible probability, where $f'$ is drawn from $F'$ uniformly at random. However, by Lemma~\ref{lem:PRPindistinguishable}, $f'(p(\cdot, \cdot))$ and a random matrix are oracle-indistinguishable; a contradiction. Therefore, $f(p(\cdot, \cdot))$ is oracle-indistinguishable from a random matrix.
\end{proof}

\subsubsection{High entropy pseudorandom matrix and $4$-wise independence}
\label{Security}

As a final point before our more technical theorems, let us state a simple but important corollary which we will use a lot in our analysis and also in our other constructions.

\begin{corollary} \label{coro:highentropy4wise}
    Let $A$ be the matrix of \eqref{A matrix}. Then $A$ is $4$-wise independent.
\end{corollary}
\begin{proof}
Note that
\begin{equation}
A_{ij}:= f(p(i, j)),
\end{equation}
where $f$ is a $4$-wise independent function and $p$ is a permutation. The proof then follows from the observation that permuting the entries of a $k$-wise independent function preserves $k$-wise independence.
\end{proof}

\subsubsection{Analysis of the Frobenius norm}
\begin{lemma}\label{lem:4wisefrobenius}
	Let $f$ be a function uniformly sampled  from a $4$-wise independent function family $F= \{f :[2^{n/2}] \times [2^{n/2}] \to \{1, -1 \} \}$, the Frobenius norm $\| \frac{1}{2^{n}} AA^{\mathsf{T}} \|_F \le 2^{-n/8}$ with high probability, where $A_{ij} := f(i, j).$
\end{lemma}
\begin{proof}
 Since entries of matrix $A$ are 4-wise independent,
	\begin{align*}
		& \mathrm{E} \left [ \left \|\frac{1}{2^{n}} AA^{\mathsf{T}} \right \|^2_F \right] \\ &= \frac{1}{2^{2n}} \mathrm{E} \left [ \left \| AA^{\mathsf{T}} \right \|_F \right] \\
		& = \frac{1}{2^{2n}} \sum_{i=1}^{2^{n/2}}  \sum_{j=1}^{2^{n/2}} \mathrm{E} \left[\left ( \sum_{k=1}^{2^{n/2}}A_{ik}\cdot A_{jk}\right)^2 \right] \\
		& = \frac{1}{2^{2n}} \sum_{i=1} ^{2^{n/2}}  \mathrm{E} \left[\left ( \sum_{k=1}^{2^{n/2}}A_{ik}\cdot A_{ik} \right)^2 \right] + \frac{1}{2^{2n}} \sum_{i \ne j, i, j = 1}^{2^{n/2}} \mathrm{E} \left[\left ( \sum_{k=1}^{2^{n/2}}A_{ik}\cdot A_{jk}\right)^2 \right] \\
		& = \frac{1}{2^{n/2}} + \frac{1}{2^{2n}} \sum_{i \ne j, i, j = 1}^{2^{n/2}} \sum_{k=1}^{2^{n/2}}\mathrm{E} \left[\left ( A_{ik}\cdot A_{jk}\right)^2 \right] \\
		& \le \frac{1}{2^{n/2-1}}.
	\end{align*}
	Finally, by the Markov's inequality, we have 
	\begin{equation}
	\Pr \left [ \left \|\frac{1}{2^{n}} AA^{\mathsf{T}} \right \|^2_F  > 2^{-n/4}\right] \le 2^{1 - n/2}.
	\end{equation}
	Therefore, \begin{equation}
	\Pr \left [ \left \|\frac{1}{2^{n}} AA^{\mathsf{T}} \right \|_F  > 2^{-n/8}\right] \le 2^{1 - n/2}.
	\end{equation}
\end{proof}
\subsubsection{Putting the two properties together}
\label{security2}

\begin{theorem}
\label{putting it together}
 Let $F= \{f :[2^{n/2}] \times [2^{n/2}] \to \{1, -1 \} \}$ be a $4$-wise independent function family, and $P= \{p: [2^{n/2}] \times [2^{n/2}] \to [2^{n/2}] \times [2^{n/2}] \}$ be a PRP. If we draw a function $f$ from $F$ uniformly at random, and $p$ from $P$ uniformly at random. Then $f(p(\cdot, \cdot))$ is oracle-indistinguishable from truly random matrix and the following hold with high probability:
 \begin{itemize}
 	\item $\| \frac{1}{2^{n}} AA^{\mathsf{T}} \|_F \le 2^{-n/8}$,  where $A_{ij} := f(p(i, j)).$
 \end{itemize}
\end{theorem}
\begin{proof}
     First, the   oracle-indistinguishability immediately comes from  Lemma~\ref{lem:4wiseindependent}. Next by Corollary~\ref{coro:highentropy4wise}, entries of $A$ are $4$-wise independent. Therefore by Lemma~\ref{lem:4wisefrobenius}, with high probability. $\| \frac{1}{2^{n}} AA^{\mathsf{T}} \|_F \le 2^{-n/8}$ holds. 
\end{proof}
Finally, with high probability,  by Jensen's inequality,
\begin{equation}
\label{eq28}
\mathsf{S}\left(\frac{1}{2^{n}} AA^{\mathsf{T}}\right) \ge -\log \left \|\frac{1}{2^{n}} AA^{\mathsf{T}} \right \|_F = \Omega(n).
\end{equation}

\subsubsection{Constructing a pseudorandom state with high entropy}
\begin{lemma}
Let $A$ be a matrix such that
\begin{equation}
\label{high entropy matrix}
    A_{i, j} = f(p(i, j)),
\end{equation}
where $F= \{f :[2^{n/2}] \times [2^{n/2}] \to \{1, -1 \} \}$ be a public $4$-wise independent function family, and $P= \{p: [2^{n/2}] \times [2^{n/2}] \to [2^{n/2}] \times [2^{n/2}] \}$ be a public PRP family. Each member of both $F$ and $P$ has an efficient description. $f$ and $p$ are drawn from $F$ and $P$ uniformly at random. Then,
\begin{itemize}
\item The state
    \begin{equation}
    \label{high entropy state construction}
        |\psi_A \rangle = \frac{1}{\sqrt{2^{n}}} \sum_{i, j \in \{0, 1\}^{n/2}} (-1)^{A(i, j)} \ket{i, j}
    \end{equation}
is a pseudorandom state with
\begin{equation}
    \mathsf{S}(\rho) = \Omega(n),
\end{equation}
where $\rho$ is the reduced density matrix on $n/2$ qubits. \item The key $\mathsf{K}$ is a binary string specifying the index of the permutation $p$ and the index of the $4$-wise independent function $f$ in the public families $P$ and $F$ respectively.
\end{itemize}
\end{lemma}
\begin{proof}
The security and entropy bounds for $\ket{\psi_A}$ already follow from Theorem \eqref{putting it together}. We will prove that when given the key $\mathsf{K}$, we can efficiently prepare $\ket{\psi_A}$. This easily follows from the fact that when given the description of two functions $f$ and $g$, the composite function $f(p(\cdot))$ can be constructed in polynomial time. So, a phase oracle for $A$ can be constructed in polynomial time from the key $\mathsf{K}$. 

Then, we just need to prepare a superposition state and query an efficiently constructed phase oracle for $A$.
\end{proof}

\subsection{Pseudorandom matrix with low entanglement entropy}
\label{subsection: low entropy}
\noindent In this section, we will give a pseudorandom matrix construction with 
\begin{equation}
\mathsf{S}\left(\frac{1}{2^{n}} BB^{\mathsf{T}}\right)= \mathcal{O}(\mathrm{polylog}(n)).
\end{equation}
Hereon, let $k = \mathrm{polylog}(n)$, for ease of notation. 
\subsubsection{High level overview of the construction}
From \eqref{rank and entropy}, if we want to upper bound the entanglement entropy of our construction, it suffices to show that the matrix $B$ has low rank. A simple way to reduce the rank of a matrix is by repeating the rows. 

We will start with a high entropy matrix $A$, which can be constructed like we saw in the previous section. We will repeat the rows of this matrix, but repeat them such that no polynomial time adversary can ``feel" that the rank has gone down. In order to introduce collisions, we will use a ``pseudorandom $2^{n-k}$-to-$1$ function" to subsample the rows, which we will construct by composing a small range pseudorandom function with another pseudorandom function. 

\subsubsection{The construction}

Let $B$ be a matrix such that
\begin{equation}
\label{lowentropymatrix}
    B_{i,j} = A_{g(i), j},
\end{equation}
where
\begin{equation}
\label{low entropy construction}
    g_{i} = f(h(i) ~\mathrm{mod} \, 2^k),
\end{equation}
with $F= \{f :[2^{n/2}] \to [2^{n/2}]\}$ and $H= \{g :[2^{n/2}] \to [2^{n/2}] \}$ being two pseudorandom function families. $f$ and $h$ are drawn from $F$ and $H$ uniformly at random, and $A$ is a high entropy pseudorandom matrix of the type in \eqref{high entropy matrix}. 
\subsubsection{Security analysis of small range distributions}
\label{small range functions}
\begin{definition}[$m$-range functions, \cite{Zha21}]
We say a function $f: [2^{n/2}] \to [2^{n/2}]$ is a $m$-range function, if the range of $f$ has at most $m$ distinct values. 
\end{definition} 
\begin{definition}[Small range distributions, \cite{Zha21}]
    Define $SR_{2^k}$ as the following distribution of function $f: [2^{n/2}] \to [2^{n/2}]$:
      \begin{itemize}
      \item For each $i \in [2^k]$, choose a random value $y_i \in [2^{n/2}]$.
      \item For each $x \in [2^n]$, pick a random $i \in [2^k]$ and set $f(x) = y_i$.
  \end{itemize}
\end{definition}

\noindent The security analysis will go through a sequence of hybrids involving this small range distribution.

\begin{lemma} \label{lem:srupperbound}
  For any $\{ 1, -1\}$-valued matrix $A$, $g: [2^{n/2}] \to [2^{n/2}]$ is a $2^k$-range function, $S(\frac{1}{2^{n}} BB^{\mathsf{T}}) = O(k)$, where $B_{i,j} = A_{g(i),j}$.
\end{lemma}

\begin{proof}
First note that $\frac{1}{2^{n}} BB^{\mathsf{T}}$ is a positive semidefinite matrix with trace $1$. The rank of $B$ is at most $2^k$, therefore by Jensen's inequality, $S(\frac{1}{2^{n}} BB^{\mathsf{T}}) \le k$.
\end{proof}

\begin{lemma}[\cite{Zha21}] \label{lem:srtrulyrandom}
	Let $g$ is a $2^k$-range function drawn from  $SR_{2^k}$. Then $g$ is oracle-indistinguishable from a truly random function for $k = \omega(n)$.
\end{lemma}

\begin{lemma} \label{lem:pseudorandomsr}
	Let $F = \{ f: [2^{n/2}] \to [2^{n/2}] \}$, $G = \{ g: [2^{n/2}] \to [2^{n/2}] \}$ be two PRFs. Then $f(g(\cdot) ~\mathrm{mod} \, 2^k)$ and $r$ are oracle-indistinguishable, where $f$ and $g$ are drawn from $F$ and $G$ respectively, and $r$ is drawn from $SR_{2^k}$. 
\end{lemma}
\begin{proof}
	This lemma is proved by a sequence of hybrids. 
	
	\paragraph{Hybrid 0.} The adversary is given the oracle $f(g(\cdot) ~\mathrm{mod} \, 2^k)$.
	\paragraph{Hybrid 1.} In this case we replace the inner PRF $G$ with a truly random function $r_1$: $f(r_1(\cdot) ~\mathrm{mod} \, 2^k)$.
	\paragraph{Hybrid 2.} In this case we replace the outer PRF $f$ with a truly random function $r_2$: $r_2(r_1(\cdot) ~\mathrm{mod} \, 2^k)$. Note that this is viable since for any $q$-query adversary $\mathcal{A}$, we can always replace the truly random function with a $2q$-wise independent function and observe exactly the same behavior.
	
	Note that in the \textbf{Hybrid 2}, the adversary is given oracle access to a truly random $2^k$-range function $r$. Thus,  $F(G(\cdot) ~\mathrm{mod} \, 2^k)$ and $SR_{2^k}$ are indistinguishable. 
\end{proof}
\subsubsection{Putting everything together}
Now we have enough tools to put everything together. We are ready to prove the last theorem. 

\begin{definition}
	Let $A$ be a $2^{n/2} \times 2^{n/2}$  $\{ 1, -1 \}$-valued matrix, $f$ be a $[2^{n/2}] \to [2^{n/2}]$ function, then the row composition of $A$ and $f$ is a $2^{n/2} \times 2^{n/2}$ $\{ 1, -1 \}$-valued matrix and is defined by $(A \circ_{\mathsf{row}} f)_{i,j} := A_{f(i), j}$.
\end{definition}

\begin{theorem}
		Let $F = \{ f: [2^{n/2}] \to [2^{n/2}] \}$, $G = \{ g: [2^{n/2}] \to [2^{n/2}] \}$ be two PRFs. Let $\mathcal{D}$  be a distribution of $2^{n/2} \times 2^{n/2}$ $\{ 1, -1 \}$-valued matrices, and oracle $O_\mathcal{D} \leftarrow \mathcal{D}$ can be efficiently constructed. Suppose $\mathcal{D}$ is oracle-indistinguishable from the uniformly random distribution $\mathcal{R}$ of $2^{n/2} \times 2^{n/2}$ $\{ 1, -1 \}$-valued matrices. Then  $A \circ_{\mathsf{row}} f(g(\cdot) ~\mathrm{mod} \, 2^k)$   is  oracle-indistinguishable from $R$, where $A$ is a matrix sampled from $\mathcal{D}$ uniformly at random, $f, g$ are drawn from $F$ and $G$ respectively, and $R$ is a truly random matrix.
\end{theorem}

\begin{proof}
	The security is proved by a sequence of hybrids.
	\paragraph{Hybrid 0.} This is the case where the adversary is given oracle access to the matrix $A \circ_{\mathsf{row}}  f(g(\cdot) ~\mathrm{mod} \, 2^k)$. 
	\paragraph{Hybrid 1.} In this case, we switch the pseudorandom small-range function $ f(g(\cdot) ~\mathrm{mod} \, 2^k)$ to be a truly random function, so that the adversary is given $A \circ_{\mathsf{row}} r$, where $r: [2^{n/2}] \to [2^{n/2}]$ is a truly random function.
	
	Suppose $\mathcal{A}$ distinguishes \textbf{Hybrid 0} from \textbf{Hybrid 1} with non-negligible probability. Then we can construct an adversary $\mathcal{B}$, which is given $A \circ_{\mathsf{row}} O$, where $O$ is either $f(g(\cdot) ~\mathrm{mod} \, 2^k)$ or atruly random function $r$. 
 
 This adversary distinguishes $f(g(\cdot) ~\mathrm{mod} \, 2^k)$ from truly random function $r$. Note that $A$ can be efficiently constructed, so that $\mathcal{B}$ is an efficient quantum algorithm. This contradicts Lemma~\ref{lem:srtrulyrandom} and Lemma~\ref{lem:pseudorandomsr} which say that $f(g(\cdot) ~\mathrm{mod} \, 2^k)$ and $r$ are indistinguishable.
	
	\paragraph{Hybrid 2.} In this hybrid, we replace pseudorandom $A$ with a truly random matrix $R$. The adversary is given $R \circ_{\mathsf{row}} r$, where $r: [2^{n/2}] \to [2^{n/2}]$ is a truly random function.
	
	Suppose $\mathcal{A}$ distinguishes \textbf{Hybrid 1} from \textbf{Hybrid 2} with non-negligible probability. Then we can construct an adversary $\mathcal{B}$, which is given $O \circ_{\mathsf{row}} r$, where $O$ is either pseudorandom $A$ or a truly random matrix $R$. This adversary distinguishes pseudorandom $A$ from a truly random matrix $R$. 
 
 Unfortunately, the random function $r$ is inefficient to construct. We solve this by using Corollary 2.2 in \cite{Zha21}, which states that such oracles can be efficiently simulated using $2q$-wise independent functions as long as the number of queries to the oracle is upper bounded by $q$ and $q$ is a polynomial in $n$. Therefore, by efficiently simulating the random function  $r$, we get an efficient algorithm to distinguishes $A$ and $R$ with non-negligible probability. This contradicts our assumption.
	 
	 \paragraph{Hybrid 3.} For this hybrid, we replace truly random function $r$ in \textbf{Hybrid 2} with truly random permutation $p$, so that the adversary is given oracle access to $R \circ_{\mathsf{row}} p$, where $p: [2^{n/2}] \to [2^{n/2}]$ is a truly random permutation.
	 
	 Suppose $\mathcal{A}$ distinguishes \textbf{Hybrid 2} from \textbf{Hybrid 3} with non-negligible probability. Then we can construct an adversary $\mathcal{B}$, which is given ${R \circ_{\mathsf{row}} O}$ that distinguishes truly random functions from truly random permutations. Note that
	 $\mathcal{B}$ can answer any query made by $\mathcal{A}$ by making a constant number of queries to its own oracle $O$ and a random function oracle $R$. By doing the same trick, we can efficiently simulate the random oracle $R$ with $2q$-wise independent functions. Thus, we have an efficient quantum algorithm distinguishes truly random functions and truly random permutations, which contradicts the indistinguishability result by Zhandry~\cite{Zha15}.
	 
	 \paragraph{Hybrid 4.} In the last hybrid, the adversary is given truly random oracles R.
	  It is easy to see that \textbf{Hybrid 3} and \textbf{Hybrid 4} are exactly the same distribution.
	  
	  In conclusion, \textbf{Hybrid 0} and \textbf{Hybrid 4} are indistinguishable, as desired.
 	
\end{proof}

Before talking about constructing pseudorandom quantum states with our pseudorandom matrix, let us first generalize the construction of our pseudorandom matrix.

\subsection{Pseudorandom matrix with tunable entanglement entropy}
In this section, we will give a pseudorandom matrix construction such that we can tune the entanglement entropy of the matrix. That is, for a given $k$, between $\omega(\log n)$ and $\mathcal{O}( n)$, we can construct a pseudorandom state with entropy $\Theta(k)$. This generalizes the low-entropy construction of the previous section.
\subsubsection{Overview of the construction}
The proof will combine techniques from both the high entropy construction and the low entropy construction. The matrix will now be of the form
\begin{equation}
\label{tunable entropy}
    C_{i,j} = A_{g(i), j},
\end{equation}
where
\begin{equation}
\label{tunable entropy construction}
    g(i) = f(h(i) ~\mathrm{mod} \, 2^k),
\end{equation}
with $F= \{f :[2^{n/2}] \to [2^{n/2}]\}$ and $H= \{g :[2^{n/2}] \to [2^{n/2}] \}$ being two pseudorandom permutation families. $f$ and $h$ are drawn from $F$ and $H$ uniformly at random, and $A$ is drawn from a $4$-wise independent pseudorandom matrix family we constructed  \eqref{high entropy matrix}. Note that this is similar to the low entropy construction in \eqref{lowentropymatrix}, with the difference being that we use pseudorandom permutations instead of pseudorandom functions, because we care about the entropy of the matrix in a more fine-grained way.

The upper bound on the entropy comes by analyzing the rank of the matrix in \eqref{tunable entropy}, similar to the analysis in \ref{subsection: low entropy}. The lower bound of the entropy comes from lower bounding  
\begin{equation}
\label{Frobenius norm}
    -\log \left(\bigg|\bigg|\frac{1}{2^n} CC^{\mathsf{T}} \bigg|\bigg|_F\right),
\end{equation}
similar to the analysis in \ref{subsection: high entropy}. To lower bound \eqref{Frobenius norm}, we show that the Frobenius norm of $\frac{1}{2^{n}} CC^{\mathsf{T}}$ is $\mathcal{O}\left(\frac{1}{2^k}\right)$. 

Finally, the security proof of the construction will go through a sequence of hybrids. First, the hybrids will replace the inner and outer pseudorandom permutations with pseudorandom functions. Then, using the same analysis as that of \ref{small range functions}, we can argue that the resultant function is indistinguishable from a random function.

\subsubsection{Upper and lower bound on the entropy}
\label{entropy tunable}

\begin{definition} 
    We say a function $f: [2^{n/2}] \to [2^{n/2}]$ is $2^{n/2-k}$-to-$1$ if the image of $f$ has exactly $2^k$ elements and the size of preimage $f^{-1}(x)$ of any $x \in [2^{n/2}]$ is either $0$ or $2^{n/2-k}$.
\end{definition}

It is clear that $f(h(\cdot) ~\mathrm{mod} \, 2^k)$ is a $2^{n/2-k}$-to-$1$ function for any two permutations $f, h : [2^{n/2}] \to [2^{n/2}]$.

\begin{lemma}\label{lem:tunablefrobenius}
Let $f$ be a function uniformly sampled from a $4$-wise independent function family $F \{f: [2^{n/2}] \times [2^{n/2}] \to \{1, -1 \} \}$, and $g: [2^{n/2}] \to [2^{n/2}]$ is a $2^{n/2-k}$-to-$1$ function. Then the Frobenius norm $\| \frac{1}{2^{n}} CC^{\mathsf{T}} \|_F \le 2^{-k/4}$ with high probability, where $C_{ij} := f(g(i), j).$
\end{lemma}

\begin{proof}
    Suppose $r_1, r_2, \ldots, r_{2^k}$ be $2^k$ distinct elements of the image of $g$. Let $B$ be a $2^k \times 2^{n/2}$ matrix such that $B_{i,j} := C_{r_i, j}$. Since $g$ is a $2^{n/2-k}$-to-$1$ function, each entry of $BB^{\mathsf{T}}$ repeats exactly $2^{n - 2k}$ times in $CC^{\mathsf{T}}$. Therefore, we have $\| \frac{1}{2^{n}} CC^{\mathsf{T}} \|^2_F = 2^{n - 2k} \cdot \| \frac{1}{2^{n}} BB^{\mathsf{T}} \|^2_F $. Note that entries in matrix $B$ are $4$-wise independent. We have,
    \begin{align*}
		& \mathrm{E} \left [ \left \|\frac{1}{2^{n}} BB^{\mathsf{T}} \right \|^2_F \right] \\ &= \frac{1}{2^{2n}} \mathrm{E} \left [ \left \| BB^{\mathsf{T}} \right \|_F \right] \\
		& = \frac{1}{2^{2n}} \sum_{i=1}^{2^k}  \sum_{j=1}^{2^k} \mathrm{E} \left[\left ( \sum_{l=1}^{2^{n/2}}B_{il}\cdot B_{jl}\right)^2 \right] \\
		& = \frac{1}{2^{2n}} \sum_{i=1} ^{2^k}  \mathrm{E} \left[\left ( \sum_{l=1}^{2^{n/2}}B_{il}\cdot B_{il} \right)^2 \right] + \frac{1}{2^{2n}} \sum_{i \ne j, i, j = 1}^{2^k} \mathrm{E} \left[\left ( \sum_{l=1}^{2^{n/2}}B_{il}\cdot B_{jl}\right)^2 \right] \\
		& = 2^{k-n} + \frac{1}{2^{2n}} \sum_{i \ne j, i, j = 1}^{2^k} \sum_{l=1}^{2^{n/2}}\mathrm{E} \left[\left ( B_{il}\cdot B_{jl}\right)^2 \right] \\
		& \le 2^{k-n} + 2^{2k - \frac32n} \\
		& \le 2^{k-n + 1}.
	\end{align*}
	
	\noindent In consequence, \[\mathrm{E}\left \|\frac{1}{2^{n}} CC^{\mathsf{T}} \right \|^2_F = 2^{n-2k}\cdot \mathrm{E}\left \|\frac{1}{2^{n}} BB^{\mathsf{T}} \right \|^2_F \le 2^{1-k}. \]
	Finally, by the Markov's inequality, we have 
	\begin{equation}
	\Pr \left [ \left \|\frac{1}{2^{n}} CC^{\mathsf{T}} \right \|^2_F  > 2^{-k/2}\right] \le 2^{1 - k/2}.
	\end{equation}
\end{proof}
 Therefore, $\left \|\frac{1}{2^{n}} CC^{\mathsf{T}} \right \|_F \le 2^{-k/4}$ with high probability. Finally, with high probability, by Jensen's inequality,
\begin{equation}
\label{eqwhatsoever}
\mathsf{S}\left(\frac{1}{2^{n}} CC^{\mathsf{T}}\right) \ge -\log \left \|\frac{1}{2^{n}} CC^{\mathsf{T}} \right \| = \Omega(k).
\end{equation}
The upper bound is quite simple, as a $2^{n/2-k}$-to-$1$ function is also a $2^k$-range function. By Lemma~\ref{lem:srupperbound} we have \[\mathsf{S}\left(\frac{1}{2^{n}} CC^{\mathsf{T}}\right) = O(k). \]

\subsubsection{Security analysis of $2^{n/2-k}$-to-$1$ functions}
\label{security tunable}

We are going to show that our construction of pseudorandom $2^{n/2-k}$-to-$1$ functions are oracle-indistinguishable from pseudorandom small range functions we have constructed for any $k = \omega(\log n)$. 

\begin{theorem}
    Let $F = \{f: [2^{n/2}] \to [2^{n/2}] \}$, and $G = \{g: [2^{n/2}] \to [2^{n/2}] \}$ be two families of pseudorandom functions. Let $P = \{p: [2^{n/2}] \to [2^{n/2}] \}$, and $H = \{h: [2^{n/2}] \to [2^{n/2}] \}$ be two families of pseudorandom permutations. Let $f, g, p, h$ be functions drawn from $F, G, P, H$ respectively uniformly at random. If $k = \omega(\log n)$ then $f(g(\cdot) ~\mathrm{mod} \, 2^k)$ is oracle-indistinguishable from $p(h(\cdot) ~\mathrm{mod} \, 2^k)$.
\end{theorem}
\begin{proof}
    The security is established by two sequences of hybrids. The first sequence of hybrids replaces the outer PRP with a PRF, and the second sequence of hybrids replaces the inner PRP with a PRF.
    
    Let's start with the first sequence of hybrids.
    \paragraph{Hybrid 0.} In this case, the adversary is given oracle access to function $p(h(\cdot) ~\mathrm{mod} \, 2^k)$.
    
    \paragraph{Hybrid 1.} This is the case where the adversary is given oracle access to $\rp(h(\cdot) ~\mathrm{mod} \, 2^k)$, where $\rp$ is a truly random permutation.
    
    Suppose $\mathcal{A}$ distinguishes \textbf{Hybrid 0} and \textbf{Hybrid 1} with non-negligible probability. We can construct an adversary  to $\mathcal{B}$ to distinguish between a PRP and a truly random permutation. $\mathcal{B}$ is given access to $\mathcal{O}$, where $O$ is either a PRP or a truly random permutation. $\mathcal{B}$ constructs $O(h(\cdot)~\text{mod}~2^k)$, and then calls $\mathcal{A}$
    
    \paragraph{Hybrid 2.} In this hybrid, the adversary is given oracle access to $\rf(h(\cdot) ~\mathrm{mod} \, 2^k)$, where $\rf$ is a truly random function.
    
     Suppose $\mathcal{A}$ distinguishes \textbf{Hybrid 1} and \textbf{Hybrid 2} with non-negligible probability. We can construct an adversary  to $\mathcal{B}$ to distinguish between a truly random permutation and a truly random function. $\mathcal{B}$ is given access to $\mathcal{O}$, where $O$ is either a truly random function or a truly random permutation. $\mathcal{B}$ constructs $O(h(\cdot) ~\mathrm{mod} \, 2^k)$, and then calls $\mathcal{A}$. This contradicts the indistinguishability result by Zhandry~\cite{Zha15}.
     
    \paragraph{Hybrid 3.} For this one, the adversary is given oracle access to $f(h(\cdot) ~\mathrm{mod} \, 2^k)$, where $f$ is sampled from a PRF family $F$ uniformly at random.
    
     Suppose $\mathcal{A}$ distinguishes \textbf{Hybrid 2} and \textbf{Hybrid 3} with non-negligible probability. We can construct an adversary  to $\mathcal{B}$ to distinguish between a PRF and a truly random function. $\mathcal{B}$ is given access to $\mathcal{O}$, where $O$ is either a truly random function or a PRF. $\mathcal{B}$ constructs $O(h(\cdot) ~\mathrm{mod} \, 2^k)$, and then calls $\mathcal{A}$. This contradicts the assumption of PRFs.
     
     We continue to replace the inner PRP with a PRF by a similar sequence of hybrids.
     
    \paragraph{Hybrid 4.} For this one, the adversary is given oracle access to $f(\rp(\cdot) ~\mathrm{mod} \, 2^k)$, where $\rp$ is a truly random permutation.

    Suppose $\mathcal{A}$ distinguishes \textbf{Hybrid 3} and \textbf{Hybrid 4} with non-negligible probability. We can construct an adversary  to $\mathcal{B}$ to distinguish between a PRP and a truly random permutation. $\mathcal{B}$ is given access to $\mathcal{O}$, where $O$ is either a truly random permutation or a PRP. $\mathcal{B}$ constructs $f(O(\cdot) ~\mathrm{mod} \, 2^k)$, and then calls $\mathcal{A}$.

    \paragraph{Hybrid 5.} For this one, the adversary is given oracle access to $f(\rf(\cdot) ~\mathrm{mod} \, 2^k)$, where $\rf$ is a truly random function.

    Suppose $\mathcal{A}$ distinguishes \textbf{Hybrid 4} and \textbf{Hybrid 5} with non-negligible probability. We can construct an adversary  to $\mathcal{B}$ to distinguish between a truly random function and a truly random permutation. $\mathcal{B}$ is given access to $\mathcal{O}$, where $O$ is either a truly random permutation or a truly random function. $\mathcal{B}$ constructs $f(O(\cdot) ~\mathrm{mod} \, 2^k)$, and then calls $\mathcal{A}$. This contradicts the indistinguishability result by Zhandry~\cite{Zha15}.

    \paragraph{Hybrid 6.} For this one, the adversary is given oracle access to $f(g(\cdot) ~\mathrm{mod} \, 2^k)$, where $g$ is sampled from a PRF family $G$ uniformly at random.

    Suppose $\mathcal{A}$ distinguishes \textbf{Hybrid 5} and \textbf{Hybrid 6} with non-negligible probability. We can construct an adversary  to $\mathcal{B}$ to distinguish between a truly random function and a PRF. $\mathcal{B}$ is given access to $\mathcal{O}$, where $O$ is either a PRF or a truly random function. $\mathcal{B}$ constructs $f(O(\cdot) ~\mathrm{mod} \, 2^k)$, and then calls $\mathcal{A}$. This contradicts the assumption of PRFs.

In conclusion, \textbf{Hybrid 0} and \textbf{Hybrid 4} are oracle-indistinguishable, as desired.
\end{proof}

We have showed that when we row-compose $f(g(\cdot) ~\mathrm{mod} \, 2^k)$ with a pseudorandom matrix, then the resulting matrix is pseudorandom as well. Since $f(g(\cdot) ~\mathrm{mod} \, 2^k)$ is oracle-indistinguishable from $p(h(\cdot) ~\mathrm{mod} \, 2^k)$, we can also get a pseudorandom matrix by composing $p(h(\cdot) ~\mathrm{mod} \, 2^k)$ with a pseudorandom matrix.

\subsubsection{Constructing a pseudorandom state with tunable entropy}
\label{construction: tunable PRS}
\begin{lemma}
Let $A$ be a matrix such that
\begin{equation}
\label{tunable entropy matrix}
    A_{i, j} = f(p(i, j)),
\end{equation}
where $F= \{f :[2^{n/2}] \times [2^{n/2}] \to \{1, -1 \} \}$ be a public $4$-wise independent function family, and $P= \{p: [2^{n/2}] \times [2^{n/2}] \to [2^{n/2}] \times [2^{n/2}] \}$ be a public PRP family. Each member of both $F$ and $P$ has an efficient description. $f$ and $p$ are drawn from $F$ and $P$ uniformly at random. Let
\begin{equation}
\label{tunable entropy matrix 2}
    C_{i,j} = A_{g(i), j},
\end{equation}
where
\begin{equation}
\label{tunable entropy construction overview}
    g_{i} = q(h(i) ~\mathrm{mod} \, 2^k),
\end{equation}
with $Q= \{q :[2^{n/2}] \to [2^{n/2}]\}$ and $H= \{g :[2^{n/2}] \to [2^{n/2}] \}$ being two public pseudorandom permutation families, each entry of which has an efficient description. $q$ and $h$ are sampled uniformly at random. $k$ is also public knowledge. Then,
\begin{itemize}
\item The state
    \begin{equation}
        |\psi_C \rangle = \frac{1}{\sqrt{2^{n}}} \sum_{i, j \in \{0, 1\}^{n/2}} (-1)^{C(i, j)} \ket{i, j}
    \end{equation}
is a pseudorandom state with
\begin{equation}
    \mathsf{S}(\rho) = \Theta(k),
\end{equation}
where $\rho$ is the reduced density matrix on $n/2$ qubits. \item 

The key $\mathsf{K}$ is a binary string having an efficient description of $A$, and the index of the two permutations $q$ and $h$.
\end{itemize}
\end{lemma}
\begin{proof}
The entropy and the security analysis follow from \ref{entropy tunable} and \ref{security tunable}.

We will prove that when given the key we can efficiently prepare $\ket{\psi_C}$. Note that given a function $f$, it is easy to construct
\begin{equation}
    t(x) = f(x)~\text{mod}~2^k,
\end{equation}
   for any $k$. So, a phase oracle for $C$ can be constructed in polynomial time from the key $\mathsf{K}$. Therefore, when given the key $\mathsf{K}$, it is easy to construct the state
 
    \begin{equation}
       |\psi_C \rangle = \frac{1}{\sqrt{2^{n}}} \sum_{i, j \in \{0, 1\}^{n/2}} (-1)^{C(i, j)} \ket{i, j},
    \end{equation}
    by preparing an equal superposition state and querying an efficiently prepared phase oracle for $C$.
\end{proof}

Note that there is nothing special about choosing the bipartition of qubits to have size $n/2$ on either side. While we do not explicitly discuss this in this paper, we observe that using almost the same arguments and security proofs, we can get a pseudorandom quantum state with tunable entanglement entropy $\Theta(k)$ for a bipartition $(A, B)$, where each partition has size $\Omega(n)$.

\section{Entanglement scaling and geometry}
Depending on how the entanglement entropy scales with respect to the geometry of the qubits, we define area and volume law scaling of entanglement entropy. Qualitatively, if entanglement scales as the size of the boundary of the bi-partition, then we call the scaling ``area-law" and if it scales as the size of the interior of the bi-partition, we call the scaling ``volume-law." 

\subsection{Definitions}
Let us introduce some notation that we will use for the rest of the paper. For an $n$--qubit state $\ket{\psi}$, let $(\mathsf{X}, \mathsf{Y})$ be any partition of the $n$ qubits. Then, for reduced density matrices $\rho_{\mathsf{X}}$ and $\rho_{\mathsf{Y}}$, let the von Neumann entropy, for each, be denoted by $\mathsf{S}(\rho_{\mathsf{X}:\mathsf{Y}})$. Let $\mathscr{H}(\C^N)$ denotes the Haar measure over $N=2^n$ basis states. 

\begin{definition}[Area--law entanglement]
\label{area-law}
     An $n$--qubit state $\ket{\psi}$ is area--law entangled if for any choice of $(\mathsf{X}, \mathsf{Y})$,
     \begin{equation}
     \mathsf{S}(\rho_{\mathsf{X}:\mathsf{Y}}) = \mathcal{O}(\B),
     \end{equation}
     where $\B$ is the size of the boundary of the partition between $\rho_{\mathsf{X}}$ and $\rho_{\mathsf{Y}}$.
\end{definition}

\begin{definition}[Volume--law entanglement]
\label{volume-law}
     An $n$--qubit state $\ket{\psi}$ is volume--law entangled if for any choice of $(\mathsf{X}, \mathsf{Y})$,
     \begin{equation}
     \mathsf{S}(\rho_{\mathsf{X}:\mathsf{Y}}) = \mathcal{O}(\V),
     \end{equation}
     where $\V = \mathsf{min}(\V_\mathsf{A}, \V_\mathsf{B})$, and $\V_\mathsf{X}$ and $\V_\mathsf{Y}$ are the sizes of the interiors of region $\mathsf{X}$ and $\mathsf{Y}$ respectively.
\end{definition}

Next, we will see how area--law pseudorandom states are impossible. Qualitatively, this is because, by appropriately choosing the right boundary, we can distinguish these states from Haar--random states, which follow volume law entropy. Since, by virtue of our construction, our pseudoentangled states are also pseudorandom, we cannot get area--law scaling for these states.

\subsection{Area--law pseudorandom states are impossible}
We will prove the following theorem.

\begin{theorem}
Let $\{\ket{\psi_k}\}$ be an ensemble of $n$-qubit quantum states, indexed by the key $k$, that follow area--law entanglement scaling. Then, for polynomially bounded $m$, there exists a distinguisher $\mathcal{D}$ such that
\begin{equation}
\left|\underset{k}{\mathsf{Pr}}\left[\mathcal{D}\left(|\psi_k\rangle^{\otimes m}\right) = 1\right] - \underset{\ket{\phi} \sim \mathscr{H}(\mathbb{C}^N)}{\mathsf{Pr}}\left[\underset{}{\mathcal{D}}\left(|\phi\rangle^{\otimes m}\right)=1\right]  \right| = \frac{1}{\poly(n)},
\end{equation}
for an appropriate choice of $\poly(n)$.
\end{theorem}

\begin{proof}
    The distinguisher $\mathcal{D}$ works as follows:
    \begin{itemize}
        \item Take two copies of the unknown state $\rho$. For each copy, parition the $n$ qubits into regions $(\mathsf{X}, \mathsf{Y})$ such that $|B|$ --- the size of the boundary of the bipartition --- is $\mathcal{O}(\log n)$ and the size of the interior of the two regions is $\omega(\log n)$. So, in other words, $|V|$, as defined in Definition \ref{volume-law}, is $\omega(\log n)$. Note that such a choice is possible for any spatial geometry. 
        \item Then, from both the copies, trace out all the qubits of one of the regions. Without loss of generality, let the region be $\mathsf{Y}$. Then, after this procedure, we are left with $\rho_{\mathsf{X}}^{\otimes 2}$.
        \item Apply the SWAP test to $\rho_{\mathsf{X}}^{\otimes 2}$.
    \end{itemize}
    When $\rho$ is a Haar random state, it follows volume--law scaling given by Definition \ref{volume-law}, which is a fact noted in previous works including \cite{Page_1993, Nahum_2017}.
    Then, according to the calculations in Proposition \ref{SWAPtest}, the SWAP test succeeds with probability
    \begin{equation}
\frac{1}{2} +\frac{1}{\omega(\text{poly}(n))},
    \end{equation}
    for any choice of $\text{poly}(n)$, with high probability over the choice of the state.
    When $\rho$ comes from $\{\ket{\psi_k}\}$, since it has area--law scaling of the form given in Definition \ref{area law}, by the same calculations the SWAP test succeeds with probability
    \begin{equation}
        \frac{1}{2} +\frac{1}{\text{poly}(n)},
    \end{equation}
    for some appropriate choice of $\text{poly}(n)$. Hence, overall we can get an inverse polynomial distinguishing bias, and the theorem follows.
\end{proof}

\subsection{Introducing quasi--area law states}

Even though area--law pseudorandom states are ruled out, in the sections below, we will outline the construction of quasi--area law pseudorandom states which are also pseudoentangled.

\begin{definition}[Quasi--area law entanglement]
\label{quasi--area law}
     An $n$--qubit state $\ket{\psi}$ is area--law entangled if for any choice of $(\mathsf{X}, \mathsf{Y})$,
     \begin{equation}
     \mathsf{S}(\rho_{\mathsf{X}:\mathsf{Y}}) = \mathcal{O}(\B \cdot \text{poly}\log n),
     \end{equation}
     where $\B$ is the size of the boundary of the partition between $\mathsf{X}$ and $\mathsf{Y}$.
\end{definition}

Finally, using a construction based on subset states, which may be of independent interest, we will construct pseudorandom pseudoentangled states that have entanglement entropy $\Theta(\text{poly}\log n)$ across every cut. Note that these states have optimally low entanglement entropy even when the size of the boundary $\B$ is extremely large --- for instance, $\B = \Omega(n)$ --- since there is no dependence of $\B$ in the scaling of the entanglement entropy. In this sense, these states are quasi--area law \emph{in the strongest possible way.}

\subsection{Quasi--area law pseudoentanglement across every cut}
We can generalize our construction in \ref{app:singlecutjls} to have pseudoentanglement across any cut, instead of a fixed bipartition. Additionally, considering two different spatial geometries, in $1$-D and $2$-D, we can construct quantum states that are quasi--area law pseudoentangled. A detailed discussion of this construction is relegated to the Appendix.

In the next section, we will construct a different pseudo-entangled state ensemble of independent interest which has optimally low pseudo--entanglement in the \emph{strongest} possible sense, regardless of the spatial geometry of the qubits.

\section{Pseudo-area law entangled PRS within the JLS phase state construction}
\label{app:arealawjls}
Our area--law construction will go through four steps.
\begin{itemize}
\item First, we state a a general property of matrices. 
\item Then, we use this property to construct a pseudorandom matrix, and a corresponding low entropy pseudorandom state in $1$--D. 
\item Then, we use entropy inequalities to show that our construction indeed follows area law. Specifically, we use the sub--additivity property of entanglement entropy.
\item Finally, we generalize our construction to $2$--D.
\end{itemize}
Finally, we will talk about efficiently preparing our state ensembles and then give a security proof using a series of hybrids.
\subsection{A property of matrices}
Consider a $2^{\m} \times 2^{\n}$ matrix $A$. Let the elements of $A$ be denoted by $A_{(\mathsf{i}, \mathsf{j})}$ where $\mathsf{i} \in \{0, 1\}^\m$, and $\mathsf{j} \in \{0, 1\}^\n$. Let $\mathsf{i_r}$ and $\mathsf{j_r}$ be the $\mathsf{r}^{\text{th}}$ bit of $\mathsf{i}$ and $\mathsf{j}$ respectively. A new matrix $A'$ is constructed as follows.
\begin{itemize}
\item For $\mathsf{t} \leq \mathsf{min}(\m, \n)$, consider a $2^{\m-\1} \times 2^{\n+\1}$ matrix $B$ whose elements $B_{(\mathsf{p}, \mathsf{q})}$ are given by
\begin{equation}
B_{(\mathsf{p}, \mathsf{q})} = A_{(\mathsf{p_1p_2 \cdots p_{\m-t}}\mathsf{q_1}\cdots \mathsf{q_{t}} ~\textbf{,}~ \mathsf{q_{t+1}}\cdots\mathsf{q_{\n}+t})},
\end{equation}
for $\mathsf{p} \in \{0, 1\}^{\m-\1}$ and $\mathsf{q} \in \{0, 1\}^{\n+\1}$.
\item Then, consider a $2^{\m-\1} \times 2^{\n+\1}$ matrix $B'$ , whose elements  $B'_{(\mathsf{r}, \mathsf{s})}$ are given by
\begin{equation}
B'_{(\mathsf{p}, \mathsf{q})} = B_{({g}\mathsf{(p)}, \mathsf{q})},
\end{equation}
where ${g}$ is a $k \rightarrow 1$ function. 
\item Finally, consider a $2^\m \times 2^\n$ matrix $A'$ whose elements $A'_{(\mathsf{i}, \mathsf{j})}$ are given by
\begin{equation}
A'_{(\mathsf{i}, \mathsf{j})} = B'_{(\mathsf{i_1 i_2 \cdots i_{\m-\1}} ~\textbf{,}~ \mathsf{i_{\m-\1+1} \cdots \mathsf{i_{\m}} \mathsf{j_1}\cdots j_{\n}})}
\end{equation}
\end{itemize}
Then, the following proposition holds.
\begin{proposition}
\label{proposition1}
$\mathsf{rank}(A') \leq \mathsf{rank}(A).$
\end{proposition}
\begin{proof}
Let there be $l$ linearly independent rows in $A$. By construction, the number of linearly independent rows in $A'$ is less than or equal to $l$, as every row of $A'$ is a row of $A$. Hence, the proposition follows.
\end{proof}

\subsection{Constructing the pseudorandom matrix}
\label{matrix}

In this section, we explicitly construct a special pseudorandom matrix of low entanglement entropy  in the following steps, which we will later harness to prepare area--law pseudoentangled states.

\subsubsection{Qualitative interpretation of the procedure}
Consider the qubits to be arranged in a line. Start somewhere in the line, and consider the pseudorandom matrix of that partition. Then, hash down its rank. Thereafter, iteratively move right, from the left, and consider the corresponding bipartitions. At each step, construct the pseudorandom matrix of that step, then "hash down" its rank. 

By Proposition \ref{proposition1}, this procedure does not result in an increase in rank during any of the iterations, and hence, does not increase the rank of the pseudorandom matrix we started from. So, for every such bipartition, the resultant pseudorandom matrix has low rank.

\subsubsection{The procedure}
\label{procedure}
The procedure goes through a series of steps.
\begin{itemize}
\item \textbf{Construct a pseudorandom matrix}: In this step, we construct a pseudorandom matrix $A$, which with respect to any possible cut has a high entanglement entropy with high probability.

Let $\mathsf{m_1} = n - \log^2 n$, and $\mathsf{m_2} = \log^2 n$. Let $A$ be a $2^{\mathsf{m_1}} \times 2^{\mathsf{m_2}}$ matrix with
\begin{equation}
A_{(\mathsf{i}, \mathsf{j})} = f(p(\mathsf{i}, \mathsf{j})),
\end{equation}
where $\mathsf{F}= \{f :[2^{\mathsf{m_1}}] \times [2^{\mathsf{m_2}}] \to \{1, -1 \} \}$ is a public $4$-wise independent function family, and $\mathsf{P}= \{p: [2^{\mathsf{m_1}}] \times [2^{\mathsf{m_2}}] \to [2^{\mathsf{m_1}}] \times [2^{\mathsf{m_2}}] \}$ is a public PRP family, and $f$ and $p$ are drawn from $\mathsf{F}$ and $\mathsf{P}$ respectively, uniformly at random. 

\item \textbf{Iteratively hash down the bond dimension}: Initialize a $2^{\mathsf{m_1}} \times 2^{\mathsf{m_2}}$ matrix $A'$, and as a base case, let $A' = A$. Set $\mathsf{t} = 0$ and $k = \log^2 n$. For each iteration $\mathsf{t}$ with $\mathsf{t} \leq n - 2\log^2 n$, we run the following procedure,

\begin{itemize}
\item  Construct a $2^{\m-\mathsf{t}} \times 2^{\n+\mathsf{t}}$ matrix $B$ whose elements are given by
\begin{equation}
B_{(\mathsf{p}, \mathsf{q})} = A'_{(\mathsf{p_1p_2 \cdots p_{\m-t}}\mathsf{q_1}\cdots \mathsf{q_{t}} ~\textbf{,}~ \mathsf{q_{t+1}}\cdots\mathsf{q_{\n}+t})},
\end{equation}
for $\mathsf{p} \in [2^{\m-\1}]$, and $\mathsf{q} \in [2^{\n+\1}]$.
\item Then, construct a matrix $B'$ such that
\begin{equation}
B'_{(\mathsf{p}, \mathsf{q})} = B_{(g_{\mathsf{t}}(\mathsf{p}), \mathsf{q})},
\end{equation}
where 
\begin{equation}
    g_{\mathsf{t}}(\mathsf{p}) = q_\mathsf{t}(r_\mathsf{t}(\mathsf{p}) ~\mathrm{mod} \, 2^k)
\end{equation}
with $\mathsf{Q}_\mathsf{t}= \{q_{\mathsf{t}} :[2^{\mathsf{m_1-t}}] \to [2^{\mathsf{m_1-t}}\}]$ and $\mathsf{R}_\mathsf{t} = \{r_\mathsf{t} :[2^{\mathsf{m_1-t}}] \to [2^{\mathsf{m_1-t}}] \}$ being two pseudorandom PRP families, and $q$ and $r$ are drawn from $\mathsf{Q}_{\mathsf{t}}$ and $\mathsf{R}_{\mathsf{t}}$ uniformly at random, respectively. Note that $g_{\mathsf{t}}$ is a $2^{\mathsf{m_1-t} - k} \rightarrow 1$ function.

\item Set 
\begin{equation}
\label{A-prime-update}
A'_{(\mathsf{i}, \mathsf{j})} = B'_{(\mathsf{i_1 i_2 \cdots i_{\m-\1}} ~\textbf{,}~ \mathsf{i_{\m-\1+1} \cdots \mathsf{i_{\m}} \mathsf{j_1}\cdots j_{\n}})},
\end{equation}
for $\mathsf{i} \in [2^{\mathsf{m_1}}]$, and $\mathsf{j} \in [2^{\mathsf{m_2}}]$.

\item Set $\mathsf{t} = \mathsf{t+1}$. 
\end{itemize}
\end{itemize}

The final output is the $2^{\mathsf{m_1}} \times 2^{\mathsf{m_2}}$ matrix $A'$ at the last iteration, which satisfies the following proposition.

\begin{theorem}
\label{security}
    $A'$ is a pseudorandom matrix.
\end{theorem}
\begin{proof}
For any iteration $0\leq t\leq n-2\log^2 n$, we use $A'_{\mathsf{t}}$ to denote the matrix obtained in Eq.~\ref{A-prime-update}. Since $q_{\mathsf{t}}$ and $r_{\mathsf{t}}$ are both pseudorandom functions, there is no efficient quantum algorithm can distinguish between $A'_{\mathsf{t}}$ and $A'_{\mathsf{t+1}}$. 

Given that the number of steps $\mathsf{t}_{\max}\leq n-2\log^2 n=\poly(n)$, we know, using a hybrid argument, that there is no efficient quantum algorithm that can distinguish between $A'_{\mathsf{t}_{\max}}$ and and $A'_{0} = A$, which is a pseudorandom matrix, by construction. Hence, $A' = A'_{\mathsf{t}_{\mathsf{max}}}$ is also a pseudorandom matrix.
\end{proof}

\begin{proposition}
\label{second proposition}
$\mathsf{rank}(A') \leq 2^{\log^2 n}$.
\end{proposition}
\begin{proof}
Follows from the inequality in Proposition \ref{proposition1} and the fact that, by virtue of construction,
\begin{equation}
    \mathsf{rank}(A) \leq 2^{\log^2 n}.
\end{equation}
\end{proof}
\subsection{Constructing the pseudorandom quantum state}
First, let us establish some notations that will make the statement of our final theorem compact. 
\subsubsection{Notations}
Let $A'$ be the matrix obtained from the procedure in Section \ref{matrix}. Let us write
\begin{equation}
\label{area law: notations}
    A' = f\left(A, g_1, g_2, \ldots, g_{n - \log^2 n}\right),
\end{equation}
where each $g_i$ is a quantum secure $2^{k}$ to $1$ function, with $k = n - \log^2 n$, and $A$ is as defined in Section \ref{matrix}. The main observation from Section \ref{matrix} is that $f$ can be implemented efficiently.

Now, consider an $n$--qubit quantum state given by:
\begin{equation}\label{eqn:1d-pseudoentanglement-state}
|\psi_{A'} \rangle = \frac{1}{\sqrt{2^{n}}} \sum_{i \in \{0, 1\}^{\mathsf{m_2}}, j \in \{0, 1\}^{\mathsf{m_1}}} (-1)^{A'(i, j)} \ket{i, j},
\end{equation}
with $\mathsf{m_1} = n - \log^2 n$, and $\mathsf{m_2} = \log^2 n$. When $A'$ is picked as in Eq.~\eqref{area law: notations}, Eq.~\eqref{eqn:1d-pseudoentanglement-state} implicitly defines an ensemble of quantum states.

Now, we are ready to state our main result.
\subsubsection{Proof of pseudorandomness}

\begin{theorem}
\label{preparation}
With high probability, the state
\begin{equation}
    |\psi_{A'} \rangle = \frac{1}{\sqrt{2^n}} \sum_{i \in \{0, 1\}^{\mathsf{m_2}}, j \in \{0, 1\}^{\mathsf{m_1}}} (-1)^{A'(i, j)} \ket{i, j},
\end{equation}
with $\mathsf{m_1} = n - \log^2 n$ and $\mathsf{m_2} = \log^2 n$, is a pseudorandom quantum state satisfying
\begin{equation}
S(\rho_{\mathsf{X}:\mathsf{Y}}) = \mathcal{O}(\poly\log n),
\end{equation}
for any contiguous partition $(\mathsf{X, Y})$ of $n$--qubits. 

The key $\mathsf{K}$ is a description of a matrix $A$ and quantum secure $2^{k}$-to-$1$ functions $\{g_1, \ldots, g_{n - \log^2 n}\}$, for $k = n- \log^2 n$, such that Eq.~\ref{area law: notations} holds.
\end{theorem}

\begin{proof}
First, let us prove that $\ket{\psi_{A'}}$ is a valid pseudorandom state ensemble. The proof has two parts:
\begin{itemize}
    \item \textbf{Efficient preparation}: When given the key $\mathsf{K}$, using the recipe in Section \ref{matrix}, $A'$ can be constructed efficiently. Then, by constructing an efficient phase oracle for $A'$, the state
    \begin{equation}
       |\psi_{A'} \rangle = \frac{1}{\sqrt{2^{n}}} \sum_{i \in \{0, 1\}^{\mathsf{m_2}}, j \in \{0, 1\}^{\mathsf{m_1}}} (-1)^{A'(i, j)} \ket{i, j},
    \end{equation}
    can be prepared efficiently by preparing an equal superposition state and querying the phase oracle.
    \item \textbf{Security}: By the hybrid argument in Section \ref{security}, $A'$ is a valid pseudorandom matrix. Hence, the corresponding phase state is a pseudorandom quantum state.
\end{itemize}
Finally, the scaling of the entanglement entropy follows from the fact that, by virtue of construction and by Proposition \ref{proposition1} and \ref{second proposition}, for any contiguous cut $(\mathsf{X}, \mathsf{Y})$,
\begin{equation}
    \mathsf{S}(\rho_{\mathsf{X}:\mathsf{Y}}) \leq \mathcal{O}(\log(\mathsf{rank}(A'))) = \mathcal{O}(\log^2 n),
\end{equation}
where the second inequality follows from Proposition \ref{second proposition}.
\end{proof}

\begin{remark}
    Quasi--area law pseudorandom states satisfying $$\mathsf{S}(\rho_{\mathsf{X}:\mathsf{Y}}) = \mathcal{O}(f(n)),$$
for any contiguous partition $(\mathsf{X, Y})$ and any function $f(n) = \omega(\log n)$ can be constructed and proven to be secure in the same way as in Section \ref{procedure} and Theorem \ref{preparation}, by appropriately tweaking the parameters.
\end{remark}

\subsection{Proving $1$--D quasi--area law}
\label{1D area law}
We are ready to prove that the pseudorandom state we just constructed follows quasi--area law in $1$--D, in terms of its entanglement scaling. 
Then the following theorem holds:
\begin{corollary}
\label{area law}
Let $|\psi_{A'}\rangle$ be the pseudorandom state from Theorem \ref{preparation}, where the $n$ qubits are arranged along a $1$-D line. Then, for any partition $(\mathsf{X}, \mathsf{Y})$ of the $n$--qubits,
\begin{equation}
\mathsf{S}\left(\rho_{\mathsf{X} : \mathsf{Y}}\right) = \mathcal{O}(\B \cdot \poly\log n),
\end{equation}
where $\B$ is the size of the boundary of the partition.
\end{corollary}
\begin{proof}
If the size of the boundary is $\B$, then by the subadditivity of the entanglement entropy, $\mathsf{S}\left(\rho_{\mathsf{X} : \mathsf{Y}}\right)$ is upper bounded by a summation of $\B$ terms, each a entanglement entropy of a contiguous partition of $n$--qubits. By Theorem \ref{preparation}, the entanglement entropy of each of these contiguous partitions is $\mathcal{O}(\poly\log n)$. 
\end{proof}

\subsection{Generalizing the construction for $2$--D}
\label{2D area law}
Consider a $\sqrt{n} \times \sqrt{n}$ $2$--D grid of qubits. As shown in Figure \ref{snake}, we fill the $2$--D grid as a $1$--D ``snake". 

\begin{corollary}
\label{area law2}
Let $|\psi_{A'}\rangle$ be the pseudorandom state from Theorem \ref{preparation}, where the $n$ qubits are arranged as a "snake" on a $2$--D grid, as shown in Figure \ref{snake}. Then, for any partition $(\mathsf{X}, \mathsf{Y})$ of the $n$--qubits,
\begin{equation}
\mathsf{S}\left(\rho_{\mathsf{X} : \mathsf{Y}}\right) = \mathcal{O}(\B \cdot \poly\log n),
\end{equation}
where $\B$ is the size of the boundary of the partition.
\end{corollary}

\begin{proof}
    Because of the geometry of the construction, a boundary of size $\B$ cuts the snake in $\mathcal{O}(\B)$ places. Each such cut defines a contiguous partition of the $n$ qubits on a $1$--D line. As in Corollary \ref{area law}, the proof then follows from the subadditivity of entanglement entropy and Theorem \ref{preparation}.
\end{proof}

\subsection{A discussion on entanglement scaling}
Note that in both Corollary \ref{area law} and Corollary \ref{area law2}, the size of the boundary $\B$ could be $\mathcal{O}(\log n)$. However, the lower bound on $\mathsf{S}\left(\rho_{\mathsf{X} : \mathsf{Y}}\right)$ is $\omega(\log n)$, by a SWAP test of the form in Proposition \ref{SWAPtest}. Hence, these constructions may have a large gap between the upper and lower bounds and, in general, the lower bound is not saturated for every cut.

In the next section, we will see a new construction where we saturate the SWAP test lower bound, for \emph{every} cut, for any spatial dimension! As we discussed before, in this sense, these states are quasi--area law in the strongest possible way.

\begin{figure}
\label{snake}
\begin{centering}
    \includegraphics[width=0.35\textwidth]{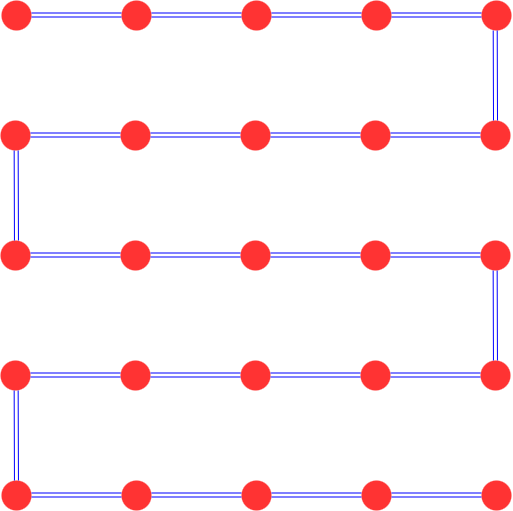}
    \caption{The blue line is how we ``snake through" a $2$--D grid.}
    \end{centering}
\end{figure}

\section{Proof that the \cite{gheorghiu2020estimating} pseudoentanglement construction is not a PRS}
\label{appendix:GHnotPRS}

\subsection{The \cite{gheorghiu2020estimating}  construction based on Extended Trapdoor Claw-Free Functions (ETCFs)}
Let
\begin{align}
    b &\in \{0, 1\}, \\
    x &\in \mathbb{Z}_{q}^{n},  \\
    s &\in \mathbb{Z}_{q}^{n}, \\
    A &\in \mathbb{Z}_{q}^{m \times n}, \\
    e_u &\in \mathbb{Z}_{q}^{m}, \\
    e' &\in \mathbb{Z}_{q}^{m} \\
    u &\in \mathbb{Z}_{q}^{m}.
\end{align}
The function \text{Gaussify}, implementable by a log depth circuit \cite{gheorghiu2020estimating}, takes $e_u$ and outputs a sample $e$ according to the truncated Gaussian distribution whose probability density function is given by
\begin{equation}
    D_{\mathbb{Z}_q, B}(x) = \frac{e^{\frac{- \pi ||x||^{2}}{B^{2}}}}{\sum_{x \in \mathbb{Z}_q^{n}, ||x|| \leq B} e^{\frac{- \pi ||x||^{2}}{B^{2}}}}.
\end{equation}
$m$ is polynomially larger than $n$ and the exponent $q$ is superpolynomially larger than $n$. For a public matrix $A$, consider the function given by \cite{brazerski, mahadev, gheorghiu2020estimating}.
\begin{equation}
    C_f(b, x, e_u) = Ax + b \cdot u + \text{Gaussify}(e_u) ~~(\text{mod}~~ q),
\end{equation}

\begin{equation}
    C_g(b, x, e_u) = Ax + b \cdot (As + e') + \text{Gaussify}(e_u) ~~(\text{mod}~~ q).
\end{equation}
$s, u,$ and $e'$ are fixed but not public. These are called extended trapdoor claw-free functions, which were studied in \cite{brazerski,  mahadev,  gheorghiu2020estimating}. Consider the following two states (normalization omitted.)
\begin{equation}
\label{first}
    \ket{\psi_{\text{C}_f}}_{A, s, e', u} \sim \sum_{x \in \mathbb{Z}_{q}^{n}, e_u \in \mathbb{Z}_{q}^{m}, b \in \{0, 1\}} \ket{b, x, e_u} \ket{C_f(b, x, e_u)},
\end{equation}
\begin{equation}
\label{second}
    \ket{\psi_{\text{C}_g}}_{A, s, e', u}  \sim \sum_{x \in \mathbb{Z}_{q}^{n}, e_u \in \mathbb{Z}_{q}^{m}, b \in \{0, 1\}} \ket{b, x, e_u} \ket{C_g(b, x, e_u)},
\end{equation}
The total number of qubits across each partition is $(n+m+1)\log q$ and $m \log q$ respectively. As shown in \cite{gheorghiu2020estimating}, with respect to the shown bipartition, assuming $A$ and $u$ are randomly sampled, and $e'$ is sampled from an appropriate Gaussian distribution, the first state  has entropy $(n+1)\log q$ approximately, and the second state has entropy $n \log q$ approximately. So, they hide approximately $\log q$ bits of entropy. 
\subsection{Sketch that neither state is a pseudorandom quantum state}
We will sketch a proof that neither $\ket{\psi_{\text{C}_f}}_{A, s, e', u}$, nor $ \ket{\psi_{\text{C}_g}}_{A, s, e', u}$ is a pseudorandom quantum state. Let us assume we are either given $m$ copies of $\ket{\psi_{\text{C}_f}}_{A, s, e', u}$ or $m$ copies of a Haar random state. The same proof works for $\ket{\psi_{\text{C}_g}}_{A, s, e', u}$. Here is a distinguisher. 

\begin{itemize}
    \item Measure sufficiently many copies of the unknown state in the standard basis.
    \item Read out the first and the second register. 
    
    \item For now, let us pretend that the state is $\ket{\psi_{\text{C}_f}}_{A, s, e', u}$. Then, the first register would give the values of $(b, x, e_u)$. Since $b$ is $0$ or $1$ with equal probability, for at least a constant fraction of the measured copies, $b = 0$. For these copies, subtract $Ax$ from the vector obtained by reading out the second register. Then, we are left with $\text{Gaussify}(e_u)$. We can verify this just by reading out the entries of the $m \times 1$ vector, plotting the histogram, and checking that it follows a Gaussian.
    \item If the state were a Haar random state, we would not observe a Gaussian histogram with high probability.
\end{itemize}

\subsection{Larger field sizes}
Consider two new functions.
\begin{equation}
\begin{aligned}
    &S_f(\overrightarrow{b}, x, e_u) = Ax + b_1 u_1 + \cdots  \\ &+b_{k-1} \cdot u_{k-1} + b_{k} \cdot u_{k} + \text{Gaussify}(e_u) ~~(\text{mod}~~ q),
    \end{aligned}
\end{equation}
\begin{equation}
\begin{aligned}
    &S_g(\overrightarrow{b}, x, e_u) = Ax + b_1 \cdot (As + e'_1) + \cdots  \\ &+b_{k-1} \cdot (As + e'_2) + b_{k} \cdot (As + e'_k) + \text{Gaussify}(e_u) ~~(\text{mod}~~ q),
    \end{aligned}
\end{equation}
where $b_1, \ldots, b_k \in \{0, 1\}$, $u_1, \ldots, u_k, e_1', \ldots, e_k' \in \mathbb{Z}_{q}^{m}$. $A$ is the public matrix. $s, u_1, \ldots, u_k, e_1', \ldots, e_k'$ are fixed but not public. $k$ could be polynomially large. Consider the superposition states \begin{equation}
\label{firstlarger}
    \ket{\psi_{\text{S}_f}}\sim \sum_{x \in \mathbb{Z}_{q}^{n}, e_u \in \mathbb{Z}_{q}^{m}, \overrightarrow{b} \in \{0, 1\}^k} \ket{\overrightarrow{b}, x, e_u} \ket{C_f(b, x, e_u)},
\end{equation}
\begin{equation}
\label{secondlarger}
    \ket{\psi_{\text{S}_g}}  \sim \sum_{x \in \mathbb{Z}_{q}^{n}, e_u \in \mathbb{Z}_{q}^{m}, \overrightarrow{b} \in \{0, 1\}^k} \ket{b, x, e_u} \ket{C_g(\overrightarrow{b}, x, e_u)}.
\end{equation}
When $A$ and $u_i$-s are randomly sampled, and $e_i'$-s are sampled from appropriate Gaussian distributions for every $i \in [k]$, the first state has entropy $(n+k) \log q$ approximately and the second state has entropy $n \log q$ approximately. So, they hide $k \log q$ bits of entropy. Let the total number of qubits in either ensemble be $t$. Then, \eqref{firstlarger} has entropy $\approx t$, and \eqref{secondlarger} has entropy $\approx ct$, for some $c < 1$. Consequently, the entropy gap is $ \approx (1-c) t$, for some constant $c < 1$.

Now, if $k$ is polynomially large, the distinguisher we sketched previously fails to distinguish either of these states from the Haar random state. This is because when \eqref{firstlarger} and \eqref{secondlarger} are measured in the standard basis, $\overrightarrow{b}$ is $0$ with an inverse exponentially small probability.

However, we sketch a proof below that neither $S_f$, nor $S_g$ is a pseudorandom function. So, to the best of our knowledge, no known technique works to show that $\ket{\psi_{\text{S}_f}}$ and $ \ket{\psi_{\text{S}_g}}$ are pseudorandom quantum states.

\subsection{Sketch that neither function is a pseudorandom function}
Consider being given black box access to either $S_f$ or $r$, where $r$ is a random function. A similar argument works when we are given $S_g$ or $r$. Let the black box be $\mathsf{B}$. 

Here is a distinguisher.
\begin{itemize}
    \item Query $\mathsf{B}(0^{k}, \mathsf{0}^{n}, e_u)$, for any choice of $e_u$, where $\mathsf{0}^{n}$ is the $n \times 1$ vector of all $0$s.
    \item We get a $m \times 1$ vector $v$. Read out the entries of $v$, and create a histogram.
    \item If the histogram is a Gaussian, say we were given $S_g$ . Else, say we were given $r$.
\end{itemize}
Note that 
\begin{equation}
    S_f(0^{k}, \mathsf{0}^{n}, e_u) = \text{Gaussify}(e_u).
\end{equation}
This is a $m \times 1$ vector. Each entry of this vector is sampled from a Gaussian distribution. Hence, when we plot the histogram, we will see a Gaussian, with high probability. This is \emph{not} true for a random function $r$. There, the histogram will look like a uniform distribution, with high probability. Hence, by looking at the plots, we can distinguish between the two distributions.

\section{Lower bound on entanglement entropy of pseudorandom states from the SWAP test}
In this section, we will prove the following proposition. A proof also follows from \cite{Ji2018}, but we include one just for the sake of completeness.

\begin{proposition}
\label{SWAPtest}
Let $\{\ket{\psi_k}\}$ be an ensemble of $n$-qubit pseudorandom quantum states, indexed by the key $k$. For a particular $\ket{\psi_k}$, consider an equipartition of the qubits of $\ket{\psi_k}$. Let $\rho_k$ be the reduced density matrix across each partition. Then,
\begin{equation}
    \mathsf{S}(\rho_k) = \omega(\log n),
\end{equation}
with high probability, where $\mathsf{S}(\rho_k)$ is the entanglement entropy of $\rho_k$.
\end{proposition}

\begin{proof}
Let us assume that the converse is true, and that
\begin{equation}
    \mathsf{S}(\rho_k) = \mathcal{O}(\log n).
\end{equation}
We will show that this implies there is a distinguisher that distinguishes $\ket{\psi}$ from a Haar random state with inverse polynomial success probability, when given two copies of one of the two states. The distinguisher acts as follows.
\begin{itemize}
    \item Start with two copies of the unknown state and trace out $n/2$ registers from each copy.
    \item Apply a SWAP test between the two reduced density matrices obtained from the first step.
    \item If the SWAP test outputs $0$, say the state was $\ket{\psi}$; else, say it was a Haar random state.
\end{itemize}
The probability of the SWAP test outputting $0$ (the success probability) is given by
\begin{equation}
\label{success probability}
    \frac{1}{2} + \frac{\Tr(\rho_{\text{unknown}}^2)}{2},
\end{equation}
where $\rho_{\text{unknown}}$ is the reduced density matrix, over $n/2$ qubits, of the unknown state. Note that
\begin{equation}
    \mathsf{S}_2(\rho_{\text{unknown}}) = -\log[\Tr(\rho_{\text{unknown}}^2)].
\end{equation}
Hence, \eqref{success probability} becomes
\begin{equation}
     \frac{1}{2} + \frac{1}{2 \cdot 2^{\mathsf{S}_2(\rho_{\text{unknown}})}}.
\end{equation}
When the unknown state is a Haar random state, from properties of a Haar random state,
\begin{equation}
     \mathsf{S}_2(\rho_{\text{unknown}}) = \mathcal{O}(n),
\end{equation}
with high probability, and hence \eqref{success probability} becomes
\begin{equation}
    \frac{1}{2} + \frac{1}{2^{\mathcal{O}(n)}}.
\end{equation}
When the unknown state is $\ket{\psi_k}$, \eqref{success probability} is
\begin{equation}
\label{success probability 2}
    \frac{1}{2} + \frac{1}{\text{poly}(n)},
\end{equation}
with high probability, which suffices for the proof of the proposition.

\noindent To see \eqref{success probability 2}, note that
\begin{equation}
    \mathsf{S}(\rho_k) = -\sum_{i} \lambda_i \log (\lambda_i),
\end{equation}
where $\{\lambda_i\}$ are elements of the spectrum of $\rho_k$, such that $\underset{i}{\sum}~ \lambda_i = 1$. Now, noting that $-\log(\cdot)$ is a convex function,
\begin{equation}
    \mathsf{S}(\rho_k) = -\sum_{i} \lambda_i \log (\lambda_i) \geq -\log\left(\sum_{i} \lambda_i^2\right) = \mathsf{S}_2(\rho_k),
\end{equation}
where the inequality in the middle follows from Jensen's inequality. Hence,
\begin{equation}
    \mathsf{S}_2(\rho_k) \leq \mathsf{S}(\rho_k) = \mathcal{O}(\log n),
\end{equation}
with high probability, from our assumption. From this observation, \eqref{success probability 2}, and hence the proof of the proposition, follows.
\end{proof}


\section{Pseudorandom states and matrix product states}
\label{matrix product states}

The following lemma holds from \cite{Vidal_2004}.
\begin{lemma}
\label{lemma11}
Consider an $n$-qubit quantum state $\ket{\psi}$ on a $1$-D line and consider contiguous partitions $(\mathsf{X}, \mathsf{Y})$ of the $n$ qubits. Let
\begin{equation}
    \chi = \underset{\mathsf{(X, Y)}}{\mathsf{max}} ~\mathsf{rank}(\rho_{\mathsf{X}:\mathsf{Y}}).
\end{equation}
Then, there exists a classical description of $\ket{\psi}$ of size $\mathcal{O}\left(n \chi^2\right).$
\end{lemma}

A natural question to ask is: how many copies of $\ket{\psi}$ do we need to find such a classical description? Moreover, when $\chi \geq 2^{\text{poly}\log n}$, the classical description is not efficient. For this case, can we do better and find an efficient classical description? 

\begin{cor}
\label{cor1}
Assuming the existence of one-way functions, there is no efficient quantum learner that outputs a polynomial sized description of a quantum state $\ket{\psi}$, even when promised that $\ket{\psi}$ has an efficient description and given access to polynomially many copies of $\ket{\psi}$.
\end{cor}

\begin{proof}
    Assume the opposite. Then, for any efficiently preparable state $\ket{\psi}$, using $\ket{\psi}^{\otimes m}$, for some polynomially bounded function $m$, we can efficiently get a classical description of $\ket{\psi}$ of polynomial size. 
    
    This, in turn, implies that we can get an efficient description of any efficiently preparable pseudorandom quantum state, and then use a SWAP test based checking procedure to break the security of the construction.
\end{proof}

\begin{remark}
Note that Corollary \ref{cor1} implies that when $\chi = 2^{\text{poly}\log n}$ and when $\ket{\psi}$ is a $1$-D quantum state, Lemma \ref{lemma11} is nearly optimal in terms of the size of the description \footnote{Even though Lemma \ref{lemma11} produces a description of size $\mathcal{O}(2^{\text{poly}\log n})$, it says nothing about runtime and number of copies. It may take superpolynomially many copies and superpolynomial runtime to learn this description.}. It produces a description of size $\mathcal{O}(2^{\text{poly}\log n})$ and there exists no efficient description of the state that is learnable with polynomially many copies.
\end{remark}

\end{document}